\documentclass[aip,jmp,reprint,onecolumn,secnumarabic,nofootinbib]{revtex4-1}

\usepackage{amsfonts,amssymb,amsmath,amsthm}
\usepackage{color}
\usepackage{array}
\usepackage{graphicx}
\usepackage{dsfont}
\usepackage{hyperref}
\usepackage{times}

\theoremstyle{definition}
\newtheorem{theorem}{Theorem}
\newtheorem{proposition}[theorem]{Proposition}

\newtheorem{lemma}[theorem]{Lemma}
\newtheorem{definition}[theorem]{Definition}
\newtheorem{corollary}[theorem]{Corollary}
\newtheorem{remark}[theorem]{Remark}

\newcommand*{\NN}{\mathbb{N}}

\newcommand*{\cB}{\mathcal{B}}

\newcommand*{\cE}{\mathcal{E}}
\newcommand*{\cF}{\mathcal{F}}

\newcommand*{\cH}{\mathcal{H}}

\newcommand*{\cK}{\mathcal{K}}

\newcommand*{\cM}{\mathcal{M}}

\newcommand*{\cN}{\mathcal{N}}
\newcommand*{\cQ}{\mathcal{Q}}

\newcommand*{\cO}{\mathcal{O}}
\newcommand*{\cP}{\mathcal{P}}
\newcommand*{\cS}{\mathcal{S}}

\newcommand*{\cT}{\mathcal{T}}

\newcommand*{\idty}{\openone}
\newcommand*{\id}{\mathrm{id}}
\newcommand*{\half}{\frac{1}{2}}
\newcommand*{\tr}{\mathrm{tr}}
\newcommand*{\ket}[1]{| #1 \rangle}
\newcommand*{\bra}[1]{\langle #1 |}

\newcommand*{\braket}[2]{\langle #1 | #2 \rangle}
\newcommand*{\proj}[1]{\ket{#1}\bra{#1}}

\newcommand\Scp[2]{\ensuremath{\, \langle #1 \,\vert #2 \,\rangle}}

\def\dmu{{{\rm d}\mu}}

\def\bra #1{\langle #1\vert}
\def\ket #1{\vert #1\rangle}
\def\braket #1#2{\langle #1 \vert #2\rangle}
\def\ketbra #1#2{\vert #1\rangle \langle #2\vert}

\def\norm #1{\Vert #1\Vert} 

\newcommand\Hmax[2]{\ensuremath{H_{\max} \left(#1 \vert #2 \right)}}
\newcommand\Hmin[2]{\ensuremath{H_{\min} \left(#1 \vert #2 \right)}}

\newcommand\dmax[2]{\ensuremath{h_{\max} \left(#1 \vert #2 \right)}}
\newcommand\dmin[2]{\ensuremath{h_{\min} \left(#1 \vert #2 \right)}}

	\definecolor{myred}{rgb}{1,0,0}

\definecolor{myblue}{rgb}{0,0,0.8}

\definecolor{myyellow}{rgb}{0.9,0.8,0}

\definecolor{mygreen}{rgb}{0,0.6,0}

\definecolor{myorange}{rgb}{0.6,0.6,0}

\definecolor{mycerul}{rgb}{0,0.6,1}

\begin{document}

\title{Position-Momentum Uncertainty Relations in the Presence of Quantum Memory}

\author{Fabian Furrer}
\affiliation{Department of Physics, Graduate School of Science,
University of Tokyo, 7-3-1 Hongo, Bunkyo-ku, Tokyo, Japan, 113-0033.}
\email[]{furrer@eve.phys.s.u-tokyo.ac.jp}

\author{Mario Berta}
\affiliation{Institute for Quantum Information and Matter, Caltech, Pasadena, CA 91125}
\affiliation{Institute for Theoretical Physics, ETH Zurich, 8093 Z\"urich.}

\author{Marco Tomamichel}
\affiliation{School of Physics, The University of Sydney, Sydney 2006, Australia}
\affiliation{Centre for Quantum Technologies, National University of Singapore, Singapore 117543.}

\author{Volkher B. Scholz}
\affiliation{Institute for Theoretical Physics, ETH Zurich, Wolfgang-Pauli-Str. 27, 8093 Z\"urich.}

\author{Matthias Christandl}
\affiliation{University of Copenhagen, Department of Mathematical Sciences, Universitetsparken 5, 2100 Copenhagen, Denmark}
\affiliation{Institute for Theoretical Physics, ETH Zurich, 8093 Z\"urich.}

\begin{abstract}
A prominent formulation of the uncertainty principle identifies the fundamental quantum feature that no particle may be prepared with certain outcomes for both position and momentum measurements. Often the statistical uncertainties are thereby measured in terms of entropies providing a clear operational interpretation in information theory and cryptography.
Recently, entropic uncertainty relations have been used to show that the uncertainty can be reduced in the presence of entanglement and to prove security of quantum cryptographic tasks.
However, much of this recent progress has been focused on observables with only a finite number of outcomes not including Heisenberg's original setting of position and momentum observables.
Here we show entropic uncertainty relations for general observables with discrete but infinite or continuous spectrum that take into account the power of an entangled observer. As an illustration, we evaluate the uncertainty relations for position and momentum measurements, which is operationally significant in that it implies security of a quantum key distribution scheme based on homodyne detection of squeezed Gaussian states.
\end{abstract}


\maketitle

\section{Introduction} 
Heisenberg's original writing~\cite{heisenberg27} allows different and sometimes conflicting
interpretations of what formalization of the uncertainty relation he had in mind; in this work we will adopt an adversarial perspective common in quantum information theory that has been fruitful in the context of quantum cryptography.
This is sometimes referred to as preparation uncertainty and goes back to Kennard~\cite{kennard27} and Robertson~\cite{robertson29}.

For the following, assume that either a position or a momentum measurement is to be performed on a quantum system prepared in an arbitrary state. An uncertainty relation then bounds the uncertainty of the measurement outcome from the perspective of an external observer (without access to the measurement device), when either the position or momentum of the particle is measured. For a particle prepared in a well-defined location, clearly the position uncertainty is low but the momentum uncertainty might be high. For a wave-like quantum system, the opposite is true. More generally, Heisenberg's uncertainty relation proclaims that, independent of the prepared state, it is impossible that both uncertainties are small.

Kennard made this statement precise by showing that the variance of the position and momentum distribution $\mathrm{Var}_{\omega}[Q]$ and $\mathrm{Var}_{\omega}[P]$ satisfy the famous inequality 
\begin{equation}
 \mathrm{Var}_{\omega}[Q] \cdot \mathrm{Var}_{\omega}[P] \geq \frac{1}{4} \, , 
\end{equation}
where throughout the paper units are in $\hbar=1$. 
Subsequently, several formulations of the uncertainty relation for different measures of uncertainty have been proposed. We are particularly interested in entropic formulations of the uncertainty principle~\cite{hirschman57,Birula75,Beckner75}, which we will discuss in more detail in Section~\ref{sec:Results}.

Nonetheless, it has been recently pointed out~\cite{Berta09} that the above picture is incomplete in the presence of quantum entanglement. In fact, it is evident that the preparation uncertainty from the perspective of an external observer can be reduced significantly if the observer is allowed to interact with the environment of the system.
For example, imagine that an approximate Einstein-Podolsky-Rosen (EPR) entangled state~\cite{epr35} is prepared\,---\,the continuous analog of a maximally entangled state\,---\,exhibiting strong correlations in both position and momentum measurements. Then the uncertainties of the outcome of position and momentum measurements are reduced if the observer is given access to a quantum memory (resp. the environment) storing one part of the approximate EPR state.
And in fact, the uncertainty simply vanishes in the limit of an ideal EPR state with perfect correlations.

Formally, this uncertainty reduction due to entanglement can best be quantified for a three party situation which reveals an interesting interplay between the uncertainty principle and monogamy of entanglement~(see, e.g., Ref.~\onlinecite{terhal04}). Suppose that two non-collaborating external observers are present, one interested in the position and the other one in the momentum measurement. Then, even though they have access to disjoint parts of the environment, the total uncertainty cannot (significantly) be reduced. To exemplify this, assume that one observer's quantum memory is in an approximate EPR state with the measured system reducing his uncertainty. Because the approximate EPR state is pure the other observer is uncorrelated. Moreover, the latter observer has high uncertainty since the measured system is in a thermal state (the partial trace of an approximate EPR state) with larger spread as the correlations of the approximate EPR state are enhanced.

In this article, we quantify this trade-off between the uncertainties of two external observers for both discretized as well as continuous position and momentum measurements. We show these relations if the uncertainties are quantified by the quantum conditional entropy, the conditional version of the von Neumann entropy, and by conditional min- and max-entropies. We also encounter the interesting phenomena for the quantum conditional entropy that the uncertainty bound is strictly lower than in the case of no quantum memory. This is in contrast to the finite-dimensional case~\cite{Berta09}.

Similar uncertainty relations have recently also been shown by Frank and Lieb~\cite{Lieb12} and by one of the present authors~\cite{FurrerPhD}. In the former work, the authors employ a more restrictive definition of the quantum conditional entropy and do not consider min- and max-entropies which are significant for quantum key distribution. 
Here, we build on the latter work and use a definition which is more suitable for systems with an infinite number of degrees of freedom. 
In particular, we introduce a general definition of the differential quantum conditional entropy in which the classical system is modeled by a $\sigma$-finite measure space and the quantum system by an arbitrary von Neumann algebra. Moreover, we show that under weak assumptions this definition is retrieved as the limit of the corresponding regularized discrete quantities along finer and finer discretization. This provides an operational approach to the differential quantum conditional entropy. We also introduce differential versions of the quantum conditional min- and max-entropies and show that they similarly emerge under weak assumptions as the regularized limits along finer and finer discretization from their discrete counterparts. 

The paper is organized as follows. First, we give a non technical discussion of the results. Then, in Section~\ref{sec:prel} we review the algebraic formalism to describe quantum and classical systems. In Section~\ref{sec:entropy} we introduce and discuss the differential quantum conditional entropies. The entropic uncertainty relations in the presence of quantum memory are proven in Section~\ref{sec:relations}. We discuss the special case of position and momentum measurements (Section~\ref{sec:posmomrelations}), and eventually conclude our results in Section~\ref{sec:outlook}).


\section{Discussion of the Results} \label{sec:Results}

In the following, we start with the differential Shannon entropy~\cite{shannon48} to measure the uncertainty associated to the outcomes of a position (or momentum) measurement. For the sake of the motivation, we first follow an operational approach and think of the differential Shannon entropy as the limit of the ordinary discrete Shannon entropy for finer and finer discretization. Consider a detector that measures if the outcome $q$ of a position measurement falls in an interval $\mathcal{I}_{k;\alpha} := (k\alpha, (k+1)\alpha]$ where $k$ is an integer and $\alpha = 2^{-n}$ the interval size for some $n \in \mathbb{N}$.
The Shannon entropy of this measurement is then 
\begin{equation}
H(Q_{\alpha})_{\omega} := - \sum_{k=-\infty}^{\infty} \omega(\mathcal{I}_{k;\alpha}) \log \omega(\mathcal{I}_{k;\alpha})\, ,
\end{equation}
 where $\omega(\mathcal{I}_{k;\alpha})$ denotes the probability of $q$ being observed in the interval $\mathcal{I}_{k;\alpha}$ when the state is $\omega$ and $Q_{\alpha}$ is the random variable indicating in which interval $q$ falls. 

A straightforward way to define the differential Shannon entropy would be to use the limit $\lim_{\alpha \to 0} [ H(Q_{\alpha})_{\omega} + \log \alpha  ]$.
The additional term $\log\alpha$ elucidates the necessity of renormalization as the
probability $\omega(\mathcal{I}_{k;\alpha})$ scales with the length of the interval, $\alpha$. However, due to ambiguities that can arise by the discretization and that one has to show the existence of the limit, it is more convenient to use the closed formula 
\begin{equation}
h(Q)_{\omega} :=  - \int \mathrm{d}q\ P_\omega(q) \log P_\omega(q)
\end{equation}
in order to define the Shannon entropy. Here $P_\omega (q)$ is the induced probability density when measuring the position of the state $\omega$. It is clear that the two definitions coincide if for instance $P_\omega$ is continuous. In the following we denote the similarly defined differential entropy for the momentum distribution by $h(P)_{\omega}$. 

Evaluating the differential entropy for Gaussian wave packets with variance $\sigma^2$ yields $h(Q)_{\omega_g} = \frac12 \log (2 \pi e \sigma^2)$ and $h(P)_{\omega_g} = \frac12 \log \frac{e \pi}{2 \sigma^2}$, respectively.  Independently, Bia{\l}ynicki-Birula and Mycielski~\cite{Birula75} and Beckner~\cite{Beckner75}showed the entropic uncertainty relation
\begin{equation}\label{eq:biaynicki}
h(Q)_{\omega} + h(P)_{\omega} \geq \log (e \pi) \, .
\end{equation}
This bound was originally derived for pure states, but also holds for mixed states (see, e.g., Ref.~\onlinecite{Birula06}). And as for Kennard's relation, Gaussian wave packets achieve equality.

These concepts can be extended to include the effects of a quantum memory. For this purpose, we first need to define a measure that characterizes the uncertainty for an observer holding quantum side information\,---\,the differential quantum conditional von Neumann entropy. Let us consider two separate physical systems, $A$ and $B$, in a joint state $\omega_{AB}$. We use the convention that $A$ is the system to be measured and $B$ is a quantum system held by an observer. A natural definition of the differential quantum conditional von Neumann entropy would again be 
 $\lim_{\alpha\to 0} \Big[ H(Q_{\alpha}|B)_{\omega} + \log \alpha \Big]\,$ ,
where $Q_{\alpha}$ has the same meaning as above and $H(Q_{\alpha}|B)_{\omega}$ is the discrete (quantum) conditional von Neumann entropy. For finite-dimensional systems as well as in the work by Frank and Lieb~\cite{Lieb12}, the conditional von Neumann entropy of a classical variable $X$ conditioned on quantum side-information $B$ is defined via the chain rule $H(X|B) = H(XB) - H(B)$.
However, this is inconvenient for our considerations here as we want to consider a quantum system $B$ that is fully general in which case both $H(XB)$ and $H(B)$ might be infinite whereas the conditional uncertainty, expressed through the quantum conditional von Neumann entropy $H(X|B)$, is still well-defined and finite. We therefore define the quantum conditional von Neumann entropy as 
\begin{equation}\label{eq:discQMCondEnt}
H(Q_{\alpha}|B)_{\omega} := - \sum_{k=-\infty}^{\infty} D \big( \omega_B^{k;\alpha} \,\big\|\, \omega_B^{\phantom{;}} \big)\ ,
\end{equation}
where $D(\cdot\|\cdot)$ is Umegaki's quantum relative entropy~\cite{umegaki62}, $\omega_B$ is the marginal state on $B$, and $\omega_B^{k;\alpha}$ is the sub-normalized marginal state on $B$ when $q$ is measured in $\mathcal{I}_{k;\alpha}$. Note the simple relation $\omega_B= \sum_k \omega_B^{k;\alpha}$. We emphasize here that in this definition, the system $B$ can be any quantum or classical system including continuous classical systems or quantized fields.
Moreover, in the finite-dimensional case the quantum relative entropy for two density matrices $\rho$ and $\sigma$ is given by $D(\rho||\sigma)=\tr\rho\log\rho -\tr\rho\log\sigma$ such that one retrieves the common definition $H(X|B) = H(XB) - H(B)$.

Similar to the case without side information, it is more convenient to define the differential conditional von Neumann entropy as a closed expression. Hence, as a natural generalization of formula~\eqref{eq:discQMCondEnt}, we define the differential quantum conditional von Neumann entropy as (see Definition~\ref{def:condneumann}) 
\begin{equation}\label{eq:closed_vN}
h(Q|B)_{\omega} = - \int_{-\infty}^{\infty} \mathrm{d} q \, D \big( \omega_B^{q} \,\big\|\, \omega_B^{\phantom{;}} \big) \, ,
\end{equation} 
where $\omega_B^q$ is the conditional state density on the system $B$ in the sense that $\int_{\mathcal{I}_{k;\alpha}} \mathrm{d} q \, \omega_B^q = \omega_B^{k;\alpha}$. In Proposition~\ref{prop:vNapprox}, we then show that if $-\infty < h(Q|B)_{\omega}$ and $H(Q_{\alpha}|B)_{\omega}<\infty$ is satisfied for an arbitrary $\alpha>0$, it follows that
\begin{equation}\label{eq:vNapprox}
h(Q|B)_{\omega} = \lim_{\alpha\to 0} \Big[ H(Q_{\alpha}|B)_{\omega} + \log \alpha \Big]\, . 
\end{equation}

Next, let us state the obtained uncertainty relations in terms of the quantum conditional von Neumann entropies. First, we consider measurements with finite spacing $\delta_q$ for $q$ and $\delta_p$ for $p$ on the A system of a tripartite state $\omega_{ABC}$.
We then find that, for all states,
\begin{equation}\label{eq:ourvn-disc}
H(Q_{\delta_q}|B)_{\omega} + H(P_{\delta_p}|C)_{\omega} \geq -\log c(\delta_q, \delta_p) \, ,
\end{equation}
with 
\begin{equation} 
c(\delta_q,\delta_p) = \max_{k,l} \Vert \sqrt{Q_k[\delta_q]}\sqrt{P_l[\delta_p]}\Vert \, 
\end{equation}
where $Q_k[\delta]$ ($P_k[\delta]$) denotes the projector onto the position (momentum) interval $\mathcal{I}_{k;\delta}$ and $\Vert\cdot\Vert$ is the operator norm. We derive the above relation in Proposition~\ref{prop:URvNdisc} for general positive operator valued measures (POVMs) for which the relation is exactly the same. Note that for projections as, e.g.,~ $Q_k[\delta _q]$ and $P_l[\delta _p]$ the square root in the complementary constant is superfluous. The complementary constant can be expressed in terms of the the $0$th radial prolate spheroidal wave function of the first kind $S_0^{(1)}$ as~\cite{Slepian1964} (see Figure~\ref{fig})
\begin{equation}
c(\delta_q,\delta_p) =  (\delta_q \delta_p)/2 \cdot S_0^{(1)}(1, ({\delta_q \delta_p})/{4})^2. 
\end{equation}

Moreover, by taking the limits $\delta_q \to 0$ and $\delta_p \to 0$, we find that
\begin{equation}\label{eq:ourvn}
h(Q|B)_{\omega} + h(P|C)_{\omega} \geq \log (2\pi) \, .
\end{equation}
Here, we have to impose the same assumptions for $P$ and $Q$ as required for the approximation in~\eqref{eq:vNapprox}. But we show that under a finite energy constraint with respect to the harmonic oscillator potential $Q^2+P^2$, it is only required that $h(Q|B)_{\omega}$ $h(P|C)_{\omega}$ are not $-\infty$. 
The relation~\eqref{eq:ourvn} is sharp in the sense that there exists a state which saturates it. In fact, we show in Section~\ref{sec:posmomrelations} that the approximate EPR state on $AB$ (or likewise $AC$) closes the gap between the left and the right hand side of~\eqref{eq:ourvn} in the limit of perfect correlations (see Figure~\ref{fig2}).

We note that the lower bound here is given by $\log (2\pi)$ and not by $\log (e\pi)$ as in the case without quantum memory~\eqref{eq:biaynicki} (see also Ref.~\onlinecite{Lieb12}). We note further that the relation without memory~\eqref{eq:biaynicki} generalizes straightforwardly to a classical memory system by the concavity of the Shannon entropy, that is, $h(Q|M) + h(P|M) \geq \log (e\pi)$ holds for every classical memory $M$. Since~\eqref{eq:ourvn} is sharp, a quantum memory can reduce the state-independent uncertainty limit. Such a phenomena has not yet been encountered in discrete variable systems and seems to be a special feature of continuous variable systems. It is related to the phenomena that for the (approximate) EPR state there exists a gap between the accessible classical correlation and the classical-quantum correlation, that is, $h(Q|Q_B) - h(Q|B)\approx \log(e/2)$. In contrast, for the maximally entangled state in the finite case---a minimal uncertainty state for mutually unbiased measurements---there exists a measurement such that the classical correlation is equal to the classical-quantum correlation.

A main motivation to use entropies to quantify the uncertainty principle stems from their operational significance. Here we consider measures important for quantum key distribution, for which the 
(differential) Shannon entropy or the quantum conditional von Neumann entropy is not suitable as it only attains operational significance when an asymptotic limit of many independent identical repetitions of a task are considered.
Instead, we extend the work by Renner~\cite{renner05} and co-workers (see Ref.~\onlinecite{tomamichel:thesis} for a recent review and Ref.~\onlinecite{Furrer10,Berta11} for infinite-dimensional generalizations) to the position-momentum setting.
In particular, we are interested in finding bounds on the optimal guessing probability of an eavesdropper. For a discretization determined by the intervals $\mathcal{I}_{k;\alpha}$ as introduced above, the guessing probability (for position) is the probability that an observer with access to the quantum system $B$ correctly predicts which interval $q$ falls into. A guessing strategy is characterized by a POVM on $B$, that is a map $k \mapsto E_B^k$ where $E_B^k$ are positive semi-definite operators in the observable algebra of $B$ that sum up to identity $\sum_k E_B^k = \idty$. Formally, we define the optimal guessing probability by
\begin{equation}\label{eq:GuessProb}
	P_{\textrm{guess}}(Q_{\alpha}|B)_{\omega} := \sup \Big\{ \sum_{k=-\infty}^{\infty} \omega_B^{k;\alpha}(E_B^k)\ \Big|\  E_B^k \textrm{ is a POVM on }B \Big\}\ .
\end{equation}
Clearly, the guessing probability is positive and at most $1$. This allows one to introduce the quantum conditional min-entropy~\cite{renner05} via the guessing probability as~\cite{koenig-2008,Furrer10} 
\begin{equation}
H_{\min}(Q_{\alpha}|B)_{\omega} := -\log P_{\textrm{guess}}(Q_{\alpha}|B)_{\omega}. 
\end{equation}

As we will see, the guessing probability is related to the decoupling fidelity defined as
\begin{equation}\label{eq:decFid}
F_{\textrm{dec}}(Q_{\alpha}|B)_{\omega} :=\sup \Big\{ \Big(\, \sum_{k=-\infty}^{\infty} \sqrt{F(\omega_B^{k;\alpha}, \sigma_B)} \Big)^{\!\!2} \ \Big|\ \sigma_B \textrm{ is a state on }B \Big\}\ ,
\end{equation}
where $F(\cdot,\cdot)$ is Uhlmann's fidelity~\cite{Uhlmann76} and we recall that $\omega_B^{k;\alpha}$ is not normalized. The decoupling fidelity, $F_{\textrm{dec}}(Q_{\alpha}|B)_{\omega}$, is a measure of how much information the marginal state on $B$ contains about $Q_\alpha$. If the state is independent, that is if $B$ does not contain any information about $Q_\alpha$, the decoupling fidelity takes its maximum value as $F_{\textrm{dec}}(Q_{\alpha}) = \big( \sum_{k=-\infty}^{\infty} \sqrt{\omega_A(\mathcal{I}_{k;\alpha})} \big)^2$. The latter expression grows with the support of the distribution, and is positive but not necessarily finite. Similarly, as in the case of the guessing probability, we can associate an entropic quantity to the decoupling fidelity, which is known as the quantum conditional max-entropy~\cite{koenig-2008,Furrer10} 
\begin{equation}
H_{\max}(Q_{\alpha}|B)_{\omega} = \log F_{\textrm{dec}}(Q_{\alpha}|B)_{\omega}. 
\end{equation}

In this work, we introduce the differential guessing probability and the differential decoupling fidelity, and accordingly, the differential quantum conditional min- and max-entropies. We define them by straightforwardly generalizing the infinite sums in~\eqref{eq:GuessProb} and~\eqref{eq:decFid} to integrals
\begin{align*}
 p_{\textrm{guess}}(Q|B)_{\omega}  := \sup \Big\{ \int_{-\infty}^{\infty} \mathrm{d}q\ \omega_B^q(E_B^q) \ \Big| \ q \mapsto E_B^q \textrm{ is a POVM on } B \Big\}\ , \\
 f_{\textrm{dec}}(Q|B)_{\omega}   :=\sup \Big\{ \Big( \int_{-\infty}^{\infty} \mathrm{d}q \sqrt{F(\omega_B^q, \sigma_B)} \Big)^2 \ \Big|\ \sigma_B \textrm{ is a state on } E  \Big\}\ . 
\end{align*}
In analogy to before, the differential quantum conditional min- and max-entropy are then defined as $h_{\min}(Q|B)_{\omega}=-\log p_{\textrm{guess}}(Q|B)_{\omega}$ and $h_{\max}(Q|B)_{\omega}=\log f_{\textrm{dec}}(Q|B)_{\omega} $ (see Definition~\ref{def:cont_entropy}). Similar as for the quantum conditional von Neumann entropy, we show that the differential quantum conditional min- and max-entropies can be retrieved in the limit of finer and finer discretization (see Proposition~\ref{thm:MinMaxApprox})
\begin{align}
h_{\min}(Q|B)_{\omega} &=\lim_{\alpha\rightarrow 0}\Big(H_{\min}(Q_{\alpha}|B)_\omega+\log\alpha\Big)\ , \label{eq:MinApprox} \\
h_{\max}(Q|B)_{\omega} &=\lim_{\alpha\rightarrow 0}\Big(H_{\max}(Q_{\alpha}|B)_\omega+\log\alpha\Big)\ , \label{eq:MaxApprox}
\end{align}
where the notation is as in~\eqref{eq:vNapprox} and~\eqref{eq:MaxApprox} holds if $H_{\max}(Q_{\alpha}|B)_\omega <\infty$ for an arbitrary $\alpha>0$.

The above quantities allow us to formulate a different entropic uncertainty relation, which generalizes the relation in~\cite{Tomamichel11} to a countable or continuous set of outcomes and general side information. For an arbitrary tripartite system in state $\omega_{ABC}$ and finite spacing $\delta_p, \delta_q$, we find that $P_{\textrm{guess}}(Q_{\delta_q}|B)_{\omega} \leq c(\delta_q,\delta_p) \cdot F_{\textrm{dec}}(P_{\delta_p}|C)_{\omega}$, where $c(\delta_q, \delta_p)$ is same as in~\eqref{eq:ourvn-disc}. Expressed in terms of entropies, this is equivalent to 
\begin{equation}\label{eq:ourmin-disc}
H_{\min}(Q_{\delta_q}|B)_{\omega} + H_{\max}(P_{\delta_p}|C)_{\omega} \leq c(\delta_q,\delta_p). 
\end{equation}
We prove the above relation for arbitrary POVM measurements in Proposition~\ref{prop:minmaxdisc}.  

The analogous relation in the continuous limit is obtained via~\eqref{eq:MinApprox} and~\eqref{eq:MaxApprox} and reads
\begin{equation}\label{eq:ourmin}
h_{\min}(Q|B)_{\omega} + h_{\max}(P|C)_{\omega} \geq \log (2\pi) \, .
\end{equation}
Note that we have to impose that $H_{\max}(P_{\alpha}|B)_\omega <\infty$ for an arbitrary $\alpha>0$. But as shown in Lemma~\ref{lem:MaxCond}, this is true if the second moments of the momentum distribution are finite. 
Both of these relations can be made sharp even in the absence of any correlation between $q$ and $B$ and  $p$ and $C$. We show in Proposition~\ref{lem:tight} that for any $\delta_p,\delta_q >0$, the discretized version~\eqref{eq:ourmin-disc} gets tight for a state with no uncertainty in the momentum degree, i.e., a state with a momentum distribution with support only on one of the intervals of length $\delta_p$. However, tightness of~\eqref{eq:ourmin} cannot be inferred from this observation because a hypothetical limit state must have an exactly defined momentum which is unphysical. But as already shown in~\cite{Birula06}, tightness of~\eqref{eq:ourmin} is given for pure Gaussian states.


\section{Preliminaries}\label{sec:prel}

\subsection{Algebraic Description of Physical Systems}\label{sec:notation}

Opposite to the standard description of quantum mechanics where the structure of the system is related to a Hilbert space, the basic objects in the algebraic approach are the observables or respectively, the algebra generated by the possible observables. It is reasonable to close the observable algebra with respect to the topology which corresponds to taking quantum mechanical expectation values, that is, the $\sigma$-weak topology. More precisely, the $\sigma$-weak topology on $\mathcal{B}(\mathcal{H})$ is the locally convex topology induced by the semi-norms $ A \mapsto | \mathrm{tr}(\tau A)| $ for trace-class operators $\tau\in\mathcal{B}(\mathcal{H})$, see Ref.~\onlinecite[Chapter 2.4.1]{Bratteli1}. Such an algebra is called a von Neumann algebra: a von Neumann algebra $\cM$ is a $\sigma$-weakly closed subalgebra of the linear, bounded operators $\cB(\cH)$ on some Hilbert space $\cH$. The algebraic approach has for instance the benefit that one can treat classical and quantum systems on the same footing. We start with specifying general quantum systems.

\paragraph*{Quantum Systems.} We associate to every quantum system a von Neumann algebra $\cM$ acting on a Hilbert space $\cH$. The set of linear, normal (i.e.~$\sigma$-weakly continuous), and positive functionals on $\cM$ is denoted by $\cP(\cM)$. The set of sub-normalized states $\cS_{\leq}(\cM)$ is defined as the elements in $\cP(\cM)$ satisfying $\omega(\idty)\leq 1$, where $\idty$ denotes the identity element in $\cM$. Elements $\omega\in\cS_{\leq}(\cM)$ with $\omega(\idty)=1$ are called (normalized) states, and the corresponding set is denoted by $\cS(\cM)$. If $\cM\cong\cB(\cH)$, then there exists a one to one correspondence between states on $\cM$ and density matrices on $\cH$. We then have for every $\omega\in\cS(\cM)$ a unique positive trace-one operator $\rho$ on $\cH$, such that for all $E\in\cM$, $\omega(E)=\tr[\rho E]$. We denote the set of density operators on $\cH$ by $\cS(\cH)$.

A multipartite system is a composite of different local subsystems $A,B,C$ associated with mutually commuting von Neumann algebras $\cM_{A},\cM_{B},\cM_{C}$ acting on the same Hilbert space $\cH$. The combined system is denoted by $\cM_{ABC}$ and is given by the von Neumann algebra generated by the individual subsystems, that is, $\cM_{ABC}=\cM_{A}\vee\cM_{B}\vee\cM_{C}$ is the $\sigma$-weak closure of the algebra $\{abc\, : \, a\in \cM_A, b\in \cM_B,c\in \cM_C\}$. If it is not clear from context, we label the correspondence of states, operators and algebras to different subsystems by lower indexes.   

By the Gelfand-Naimark-Segal (GNS) construction every $\omega\in\cS_{\leq}(\cM)$ admits a purification that is a triple $(\cK,\pi,\xi_\omega)$, $\cK$ being a Hilbert space, $\pi$ a representation of $\cM$ on $\cK$, and a sub-normalized vector $\xi_\omega \in \cH$ such that $\omega(x)=\Scp{\xi_\omega}{\pi(x)\xi_\omega}$ for all $x\in\cM$  (see, e.g., Ref.~\onlinecite[Chapter 2.3.3]{Bratteli1}). We often speak of the commutant $\pi(\cM)'$ of $\pi$ on $\cK$ as the purifying system.

The space $\cP(\cM)$ can be equipped with two different, albeit equivalent notions of distance~\cite{Bures69,Uhlmann76}. The first one is the usual norm induced by$\cM$ and defined for $\omega\in\cP(\cM)$ as
\begin{align}
\|\omega\|=\sup_{E\in\cM,\|E\|\leq1}|\omega(E)|^{2}\ .
\end{align}
For $\cM=\cB(\cH)$ and density matrices this corresponds to the usual trace-distance. The second one is called the fidelity and was introduced by Uhlmann \cite{Uhlmann76}. The fidelity for $\omega,\sigma\in\cS_{\leq}(\cM)$ is defined as
\begin{align}\label{eq:fidelity}
F(\omega,\sigma) =\sup |\Scp{\xi_\omega}{\xi_\sigma}|^{2}\ ,
\end{align}
where the supremum runs over all purifications of $\omega$ and $\sigma$ being defined with respect to the same Hilbert space. This is a non-empty set since there exists a Hilbert space $\cK$ and a representation $\pi$ of $\cM$ on $\cK$, called standard form, such that every state on $\cM$ has a purification on $\cK$, Ref.~\onlinecite[Chapter 9]{Takesaki2}.\\


\paragraph*{Classical Systems.}

A classical system is specified by the property that all possible observables commute, and can thus be described by an abelian von Neumann algebra. This perspective allows one to use the same definitions for states on classical systems as defined for quantum systems in the previous paragraph. Since classical systems will play a major role in the sequel, we discuss them in more detail.

For the sake of illustration, we start with countable classical systems denoted by $X$. In the following, we denote quantum systems by indexes $A,B,C$ and classical systems by indexes $X,Y,Z$. In the classical case we use $X,Y,Z$ to specify the subsystem as well as the range of the classical variable. The von Neumann algebra corresponding to a countable classical system $X$ is $\ell^\infty(X)$, that is, the set of functions from $X$ to $\mathbb C$ equipped with the supremum norm. Here, one can think of $e_x=(\delta_{x,k})_k$ as the measurement operator corresponding to the outcome $x\in X$. A classical state is then a normalized positive functional $\omega_X$ on $\ell^{\infty}(X)$, which can be identified with a probability distribution on $X$, that is, $\omega_X\in \ell^1(X)$. It is often convenient to embed the classical system $\ell^\infty(X)$ into the quantum system with Hilbert space dimension $X$ as the algebra of diagonal matrices with respect to a fixed basis $\{\ket x \}_{x\in X}$. A classical state $\omega_X$ can then be represented by a density operator
\begin{align}\label{eq:rhoX}
\rho_{\omega_X}=\sum_x \omega_X(x) \ketbra xx \ ,
\end{align}
such that the probability distribution can be identified with the eigenvalues of the corresponding density operator.

Let us now go a step further and consider classical systems with continuous degrees of freedom. In order to define such systems properly, we start with $(X,\Sigma,\mu)$ a measure space with $\sigma$-algebra $\Sigma$, and measure $\mu$. In the following, we will always assume that the measure space is $\sigma$-finite. The von Neumann algebra of the system is given by the essentially bounded functions denoted by $L^\infty(X)$. A classical state on $X$ is defined as a normalized positive and normal functional on $L^\infty(X)$, and may be identified with an element of the integrable complex functions $L^1(X)$, which is almost everywhere non-negative and satisfies
\begin{align}
\int_X \omega_X(x)  \dmu (x) = 1\ .
\end{align}
Such functions in $L^1(X)$ are also called probability distributions on $X$. The most prominent example of a continuous classical system is $X=\mathbb R$ with the usual Lebesgue measure. This is of course the relevant classical system in the case of position or momentum measurement.

Note that the case of a discrete classical system is obtained if the measure space $X$ is discrete, and equipped with the equally weighted discrete measure $\mu(I)=\sum_{x\in I} 1$ for $I\subset X$. In the discrete case,~\eqref{eq:rhoX} defines a representation of a classical state as a diagonal matrix of trace-one on the Hilbert space with dimension equal to the classical degrees of freedom. However, in the case of continuous variables this representation is not possible if we demand that the image is a valid density operator. This is easily seen from the fact that every density operator is by definition of trace class, and hence, has discrete spectrum. 


\paragraph*{Classical-Quantum Systems.}\label{sec:CQsyst}

Let us take a closer look at bipartite systems consisting of a classical part $X$ and a quantum part $B$. For a countable classical part $X$, the combined system is described by the von Neumann algebra (see, e.g., Ref.~\onlinecite[Chapter 6.3]{Murphy})
\begin{align}
\cM_{XB}=\ell^\infty(X)\vee\cM_B \cong\ell^\infty(X)\otimes\cM_B
 \cong\ell^\infty(X,\cM_B)=\{f:X\rightarrow\cM_B:\,\sup_x\Vert f(x)\Vert \leq \infty \}\ .
\end{align}
A state on $\cM_{XB}$ is called a classical-quantum state and can be written as $\omega_{XB}=(\omega_B^x)_{x\in X}$ with $\omega_{B}^x\in\cS_{\leq}(\cM_B)$ and $\sum_x \omega_B^x(\idty)=1$. If the quantum system $B$ is given by the set of all bounded linear operators on a Hilbert space $\cH_B$, we can represent $\omega_{XB}$ uniquely by the density operator
\begin{align}
\rho_{\omega_{XB}}=\sum_x \ketbra xx \otimes\rho_{\omega_B^x}\ .
\end{align}
It is now straightforward to generalize the above introduced classical-quantum systems from countable to continuous classical systems. The combined system is then described by the von Neumann algebra (see, e.g., Ref.~\onlinecite[Chapter 6.3]{Murphy})
\begin{align}
\cM_{XB}=L^\infty(X)\vee\cM_B\cong L^\infty(X)\otimes\cM_B
\cong L^\infty(X,\cM_B)\ ,
\end{align}
where $L^\infty(X,\cM_B)$ denotes the space of essentially bounded functions with values in $\cM_B$. The normal, positive functionals on $\cM_{XB}$ are given by elements in $L^1(X,\cP(\cM_B))$, and states can be identified with integrable functions $\omega_{XB}$ on $X$ with values in $\cP(\cM_B)$ satisfying the normalization condition 
\begin{align}
\int_X \omega_{B}^x(\idty) \dmu (x) = 1\ .
\end{align}
In analogy to the discrete case, we write the argument of the map $\omega_{XB}$ as an upper index. The evaluation of $\omega_{XB}$ on an element $E_{XB}\in L^\infty(X,\cM_B)$ is computed by $\omega_{XB}(E_{XB}) = \int_X \omega_{B}^x(E_B(x)) \dmu (x) $. For further details we refer to Ref.~\onlinecite[Chapter 4.6-4.7]{Takesaki1}. \\


\subsection{Channels, Measurements, and Post-Measurement States}\label{sec:observables}

We call an evolution of a system a channel. As we work with von Neumann algebras it is convenient to define channels as maps on the observable algebra, which is also called the Heisenberg picture. A channel from system $A$ to system $B$ described by von Neumann algebras $\cM_{A}$ and $\cM_{B}$, respectively, is given by a linear, normal, completely positive, and unital map $\cE:\cM_{B}\rightarrow\cM_{A}$. A linear map $\Phi:\cN \to \cM$ between two von Neumann algebras is called completely positive, if the map $(\id_n \otimes \Phi): M_n \otimes \cN \to M_n \otimes \cM$ is positive for all $n \in \mathbb{N}$. The map is called unital, if $\phi(\idty_{N})=\idty_{M}$. Note that $\cM_A$ and $\cM_B$ can be either a classical or a quantum system. If both systems are quantum (classical), we call the channel a quantum (classical) channel.

A measurement with outcome range $X$ is a channel which maps $L^{\infty}(X)$ to a von Neumann algebra $\cM_{A}$. Its predual then maps states of the quantum system $A$ to states of the classical system $X$. We denote the set of all measurements $E:L^{\infty}(X)\rightarrow \cM_A$ by $\rm{Meas}(X,\cM_A)$. If $X$ is countable, we can identify a measurement $E:\ell^\infty(X) \rightarrow \cM_A$ by a collection of positive operators $E_x=E(e_x)$ ($x\in X$) satisfying $\sum_x E_x = \idty$ (we denote by $e_x$ the sequence with $1$ at position $x$ and $0$ elsewhere). More generally, given a $\sigma$-finite measure space $(X,\Sigma,\mu)$ and the associated algebra $L^\infty(X)$, the mapping $\cO \to \chi_\cO \to \cE(\chi_\cO)$, for $\cO \in \Sigma$, $\chi_\cO$ being its indicator function, defines a measure on $X$ with values in the positive operators of $\cM_A$. Note that therefore our definition coincides with usual definition of a measurement as a positive operator valued measure, Ref.~\onlinecite[Chapter 3.1]{Davies:1970ux}. We define the post-measurement state obtained when measuring the state $\omega_A\in \cS(\cM_A)$ with $E_{X}\in\rm{Meas}(X,\cM_A)$ by the concatenation $\omega_X=\omega_A \circ E_{X}$, that is, $\omega_X(f) = \omega_A(E_X(f))$ for $f\in L^{\infty}(X)$. Since $\omega_A$ and $E_{X}$ are normal, so is $\omega_X$, such that $\omega_X$ is an element of the predual of $L^\infty(X)$, which is $L^1(X)$. Hence, the obtained post measurement state is a proper classical state and can be represented by a probability distribution on $X$.

In the following, we are particularly interested in the situation where we start with a bipartite quantum system $\cM_{AB}$, and measure the $A$-system with some $E_{X}\in\rm{Meas}(X,\cM_A)$. The post-measurement state is then given by $\omega_{XB}=\omega_{AB}\circ E_{X}$. Similarly as in the case of a trivial $B$-system, one can show that the state $\omega_{XB}$ is a proper classical-quantum state on $L^\infty(X)\otimes\cM_B$ as introduced in the previous paragraph.


\subsection{Discretization of Continuous Classical Systems}\label{sec:discretisation}

Let us consider a classical system $L^\infty(X)$ with $(X,\Sigma,\mu)$ a $\sigma$-finite measure space, where $X$ is also equipped with a topology. The aim is to introduce a discretization of $X$ into countable measurable sets along which we later show the approximation of the differential quantum conditional von Neumann entropy and the differential quantum conditional min- and max-entropy.

We call a countable set $\cP=\{I_k\}_{k\in \Lambda}$ ($\Lambda$ any countable index set) of measurable subsets $I_k\in\Sigma$ a partition of $X$ if $X=\bigcup_k I_k$, $\mu(I_k \cap I_l) = \delta_{kl}\cdot\mu(I_k)$, $\mu(I_{k})<\infty$, and the closure $\bar I_k$ is compact for all $k\in\Lambda$. If $\mu(I_k)=\mu(I_l)$ for all $k,l\in\Lambda$, we call $\cP$ a balanced partition, and denote $\mu(\cP)=\mu(I_k)$. Note that the property $\mu(I_k \cap I_l) = \delta_{kl}\cdot\mu(I_k)$ implies that the step functions associated to a fixed partition form a subalgebra of all essentially bounded functions on $X$. If for two partitions $\cP_1$, $\cP_2$ all sets of $\cP_1$ are subsets of elements in $\cP_2$, we say that $\cP_2$ is finer than $\cP_1$ and write $\cP_2 \leq \cP_1$. 
A family of partitions $\{\cP_\alpha\}_{\alpha\in\Delta} $ with $\Delta$ a discrete index set approaching zero such that each $\cP_\alpha$ is balanced, $\cP_\alpha \leq \cP_{\alpha'}$ for $\alpha \leq \alpha'$, $\mu(\cP_\alpha) = \alpha$, and $\bigcup_\alpha \cP_\alpha$ generates $\Sigma$, is called an ordered dense sequence of balanced partitions. For simplicity, we usually omit the index set $\Delta$. 
\begin{definition}
 We call an ordered dense sequence of balanced partitions $\{\cP_\alpha\}$ of a measure space $(X,\Sigma,\mu)$ a coarse graining of $X$.
A quadruple $(X,\Sigma,\mu,\{\cP_\alpha\})$ is called a coarse grained measure space if the measure space is $\sigma$-finite and $\{\cP_\alpha\}$ is a coarse graining of $X$.
\end{definition}
Note that not every $\sigma$-finite measure space admits a coarse graining in the sense of the above definition. As a simple example of a measure space that admits a coarse graining consider a discrete space with the counting measure where each partition consists of sets with measure at least one. 
In the case that $X = \mathbb{R}$, $\Sigma$ the Borel $\sigma$-algebra, and $\mu$ the Lebesgue measure, a coarse graining can be easily constructed. For a positive real number $\alpha$, let us take a partition $\cP_{\alpha}$ of $\mathbb R$ into intervals $I_{k}= [ k\alpha , (k+1)\alpha]$, $k\in\mathbb Z$, with $\mu(\cP_{\alpha})=\alpha$ as introduced in Section~\ref{sec:Results}. Choosing for $\alpha$ the sequence $\frac{1}{2^n}$, $n \in \mathbb{N}$ then gives rise to a coarse graining. 
\begin{remark} 
Every Lebesgue measurable subset $X\subset \mathbb R$ equipped with the Lebesgue measure restricted to $X$ admits a coarse graining. 
\end{remark}

For a classical-quantum system $\cM_{XB}=L^\infty(X)\otimes \cM_B$, and a partition $\cP=\{I_k\}_{k\in\Lambda}$ of $X$, we can define the discretized state $\omega_{X_\cP B}\in\cS(\ell^{\infty}(\Lambda)\otimes\cM_{B})$ of $\omega_{XB}\in\cS(\cM_{XB})$ by
\begin{align}\label{eq:DiscState}
\omega_{X_{\cP}B}\big((b_{k})\big) = \sum_{k\in\Lambda}\int_{I_k} \omega_B^x(b_{k}) \, d\mu(x)=\sum_{k\in\Lambda}\omega_{B}^{\cP,k}(b_{k})\ ,
\end{align}
where $(b_{k})\in\ell^{\infty}(\Lambda)\otimes\cM_{B}$. The new classical system induced by the partition is denoted by $X_\cP$ and it is clear that $X_\cP \cong \Lambda$. In a similar way we define the discretization of a measurement $E\in\rm{Meas}(X,\cM_A)$ with respect to a partition $\cP=\{I_k\}_{k\in\Lambda}$ as the element $E^{\cP}\in\rm{Meas}(X_\cP,\cM_A)$ determined by the collection of positive operators
\begin{align}\label{eq:DiscMeasurement}
E^{\cP}_k = E(\chi_{I_k})\ ,
\end{align}
where $\chi_{I_k}$ denotes the indicator function of $I_k$. Note that the concept of a discretized measurement and a discretized state are compatible in the sense that the post-measurement state obtained from the discretized measurement $E^\cP$ is equal to the one which is obtained where one first measures $E$ and then discretizes the state. Hence, we have that $\omega_{X_{\cP}B}=\omega_{AB}\circ E^\cP$ if $\omega_{XB}=\omega_{AB}\circ E$.



\section{Quantum Conditional Entropy Measures}\label{sec:entropy}


\subsection{Quantum Conditional von Neumann Entropy}\label{sec:vNentropy}

In order to motivate our definition of the differential quantum conditional von Neumann entropy, let us first recall the situation for discrete finite classical systems and finite-dimensional Hilbert spaces. For a classical-quantum density operator $\rho_{XB}=\sum_{x}p_{x}\proj{x}_{X}\otimes\rho_{B}^{x}$, the conditional von Neumann entropy is defined as $H(X|B)_{\rho}=H(XB)_{\rho}-H(B)_{\rho}$, where $H(XB)_{\rho}=-\mathrm{tr}[\rho_{XB}\log\rho_{XB}]$ denotes the von Neumann entropy. In the following, we use that the conditional von Neumann entropy can also be rewritten as
\begin{align}\label{eq:relative_finite}
H(X|B)_{\rho}= - \sum_{x} \tr\big[ p_{x}\rho_{B}^{x} \,(\log p_{x}\rho_{B}^{x} -\log \rho_{B})  \big] = -\sum_{x}D(p_{x}\rho_{B}^{x}\|\rho_{B})\ ,
\end{align}
where the quantum relative entropy of two density matrices $\rho$ and $\sigma$ acting on a finite-dimensional Hilbert space $\cH$ is defined as (see, e.g., Ref.~\onlinecite{PetzQI})
\begin{align}
D(\rho\|\sigma)=\mathrm{tr}[\rho\log\rho]-\mathrm{tr}[\rho\log\sigma]\ ,
\end{align}
in the case where the support of $\rho$ is contained in the support of $\sigma$, and $\infty$ else.
Writing the conditional von Neumann entropy in terms of the quantum relative entropy is motivated by the fact that the latter has a well behaved extension to states on von Neumann algebras which was introduced by Araki~\cite{Araki76} and further studied by Petz and various authors (see Ref.~\onlinecite{PetzBook} and references therein). This generalization can be understood in the finite-dimensional case by writing 
\begin{align}
D(\rho\|\sigma)= \tr\Big[\rho^{1/2}\log\big(\Delta(\rho/\sigma)\big)\rho^{1/2}\Big]\ ,
\end{align}
where the so-called spatial derivative is defined as $\Delta(\rho/\sigma)=L(\sigma^{-1})R(\rho)$, where $L(a)$ and $R(a)$ denote the left and right multiplication by an element $a\in \cB(\cH)$, respectively. Here, $\sigma^{-1}$ denotes the pseudo inverse on the support of $\sigma$. Note that $\Delta(\rho/\sigma)$ defines  a linear positive operator acting on the Hilbert space $\mathrm{HS}(\cH)$ of Hilbert-Schmidt operators on $\cH$. We emphasize that the mapping $\pi: a \mapsto L(a)$, $a \in \cB(\cH)$ defines a representation of the algebra $\cB(\cH)$ on the Hilbert space $\mathrm{HS}(\cH)$. Before discussing the spatial derivative on von Neumann algebras we first consider its properties in the case of density operators (see also Ref.~\onlinecite[Chapter 3.4]{PetzQI}). The spatial derivative may then be defined by the quadratic form
\begin{align}\label{eq:15}
q: a \mapsto \tr\Big[ \rho R(\sigma^{-\half} \, a) R(\sigma^{-\half}\, a)^*\Big] = \tr\left[\rho a^* \sigma^{-1} a \right]\ ,
\end{align}
where again $R(\sigma^{-\half} \, a)$ is the right multiplication by $\sigma^{-\half} \, a$. The sesquilinear form $s: (a,b) \mapsto \tr\left[a^* \Delta(\rho/\sigma)(b) \right]$ defining the positive linear operator $\Delta(\rho/\sigma)$ is derived from $q$ by the polarization identity $s(a,b) = \frac{1}{4}(q(a+b) - q(a-b) + i q(a-ib) -i q(a+ib))$. The operator $R(\sigma^{-\half} \, a) R(\sigma^{-\half}\, a)^*$ commutes with all operators acting by left multiplication, and hence is an element of the commutant of $\pi(\cB(\cH))$. This characterization of the spatial derivative can be generalized to states on von Neumann algebras. 

For a von Neumann algebra $\cM\subseteq\cB(\cH)$, let $(\xi_\sigma, \pi_\sigma,\cH_\sigma)$ be the GNS-triple associated with $\sigma\in\cP(\cM$). For vectors $\eta$ in the set $\{\,\eta \in \cH\,:\, \norm{a \eta} \leq c_\eta \sigma(a^* a),\, a \in \cM,\,c_\eta > 0 \,\}$ with closure equal to the support of $\sigma$, we may define a linear bounded operator from $\cH_\sigma$ to $\cH$ by
\begin{align}\label{eq:14}
r_\sigma(\eta) \,:\,x\xi_\sigma \mapsto x \eta\ .
\end{align}
Note that the GNS construction ensures that the linear span of vectors of the form $x\xi_\sigma$, $x \in \cM$ are dense in $\cH_\sigma$. If the Hilbert spaces $\cH_\sigma$ and $\cH$ are isomorphic and $\eta = c x\xi_\sigma$ for $c \in \cM'$, then $r_\sigma(\eta) = c$. Moreover, the operator $r_\sigma(\eta) r_\sigma(\eta)^*$ is always an element of $\cM'$. Let now $\omega$ be a state on $\cM$, which is implemented by a vector $\xi \in \cH$, that is, $\xi$ is a purifying vector of $\omega$. The vector $\xi$ also defines a state $\omega'_\xi$ on the commutant $\cM'$ by $\omega'_\xi(y) = \Scp{\xi}{y \xi}$ for $y \in \cM'$. We define the spatial derivative $\Delta(\omega'_\xi/ \sigma)$ as the self-adjoint operator associated with the quadratic form on $\cH$ given by
\begin{align}\label{eq:16}
q: \eta \mapsto \omega'_\xi(r_\sigma(\eta) r_\sigma(\eta)^*)\ .
\end{align}
For a detailed derivation of its properties, see Ref.~\onlinecite[Chapter 9.3]{Takesaki2}, and Ref.~\onlinecite[Chapter 4]{PetzBook} and references therein. In analogy with the finite-dimensional case, we can now define the quantum relative entropy in terms of this operator (following Araki~\cite{Araki76}).

\begin{definition}\label{def:relative}
Let $\cM\subseteq\cB(\cH)$ be a von Neumann algebra acting on a Hilbert space $\cH$, $\omega\in\cS(\cM)$ implemented by a vector $\xi \in \cH$, and $\sigma\in\cP(\cM)$. If $\xi$ is in the support of $\sigma$, then the quantum relative entropy of $\omega$ with respect to $\sigma$ is defined as
\begin{align}\label{eq:relative}
D(\omega\|\sigma)=  \bra{\xi}\log\left(\Delta(\omega'_\xi/\sigma)\right)\xi\rangle\ .
\end{align}
The logarithm of the possibly unbounded operator $\Delta(\omega'_\xi/\sigma)$ is defined via the functional calculus. If $\xi$ is not in the support of $\sigma$, we set $D(\omega\|\sigma)=\infty$. 
\end{definition}

It can be shown that the quantum relative entropy is independent of the particular choice of some purifying vector $\xi$ of $\omega$, see the discussion in Ref.~\onlinecite[Chapter 5]{PetzBook} together with~\cite{Uhlmann:1977uo}. We define now the differential quantum conditional von Neumann entropy as the integral version of~\eqref{eq:relative_finite}. For later purposes, we also include an additional finite-dimensional quantum system.

\begin{definition}\label{def:condneumann}
Let $\cM_{XAB}=L^{\infty}(X)\otimes\cB(\cH_{A})\otimes\cM_{B}$ with $(X,\Sigma,\mu,)$ a $\sigma$-finite measure space, $\cH_{A}$ a finite-dimensional Hilbert space, $\cM_{B}$ a von Neumann algebra, and $\omega_{XAB}\in\cS_{\leq}(\cM_{XAB})$. Then, the conditional von Neumann entropy of $XA$ given $B$ is defined as
\begin{align}
h(XA|B)_{\omega}=-\int D(\omega_{AB}^{x}\|\tr_A \otimes \omega_{B})\, d\mu(x)\ .
\end{align}
where $\tr_{A}$ is the trace on $\cH_{A}$. 
\end{definition}
In the following, we use lower case letter $h(X|B)_{\omega}$ if $X$ is a continuous measure space and use uppercase letter, $H(X|B)_{\omega}$, if $X$ is discrete. In the latter case of a discrete measure space, we recover the formula in~\eqref{eq:relative_finite} with now a possible infinite sum 
\begin{align}
H(X|B)_{\omega}=  -\sum_{x\in X}D(p_{x}\omega_{B}^{x}\|\omega_{B})\ .
\end{align}

The following statement shows that the differential quantum conditional von Neumann entropy can be retrieved from the regularized version of its discrete counterpart in the limit of finer and finer coarse grainings. 
\begin{proposition}\label{prop:vNapprox}
Let $\cM_{XB}=L^{\infty}(X)\otimes\cM_{B}$ with $\cM_{B}$ a von Neumann algebra and $(X,\Sigma,\mu,\{\cP_\alpha\})$ a coarse grained measure space.  Consider $\omega_{XB}\in\cS(\cM_{XB})$ such that $-\infty<h(X|B)_{\omega}$, and assume that there exists $\alpha_0> 0$ for which $H(X_{\cP_{\alpha_0}} |B)_{\omega} < \infty$. Then, it follows that
\begin{align} 
h(X|B)_{\omega}=\lim_{\alpha\rightarrow 0} \left(H(X_{\cP_{\alpha}} |B)_{\omega}+\log\alpha \right)\ ,
\end{align}
where $\omega_{X_{\cP_{\alpha}}B}$ is defined as in~\eqref{eq:DiscState}. Furthermore, if $h(X)_{\omega}<\infty$, then it follows that
\begin{align}\label{eq:kuznetsova}
h(X|B)_{\omega}=h(X)_{\omega}-D(\omega_{XB}\|\omega_{X}\otimes\omega_{B})\ .
\end{align}
\end{proposition}
We note that in Ref.~\onlinecite{Kuznetsova10}, the conditional von Neumann entropy was defined as in~\eqref{eq:kuznetsova} for $\omega_{AB}\in\cS\big(\cB(\cH_{A}\otimes\cH_{B})\big)$ with $H(A)_{\omega}<\infty$, and separable Hilbert spaces $\cH_{A}$ and $\cH_{B}$.
\begin{proof}
We first write the integral as a series of integrals over a covering $\{X^k\}_{k=0}^\infty$ of $X$ by compact measurable sets with $\mu(X^k\cap X^l) = 0$ for $k\neq l$. Using that the Lebesgue integral can be split over positive and negative parts of the integrand, we can use the monotone convergence theorem to obtain
\begin{align}
 - h(X|B)_{\omega} = \lim_{n\rightarrow \infty}  \sum_{k=1}^n \int_{X^k} D(\omega_{B}^{x}\| \omega_{B})\, d\mu(x)\ .
\end{align}
For a fixed $k$, it follows from disintegration theory, Ref.~\onlinecite[Chapter IV.7]{Takesaki1}, that 
\begin{align}\label{eq:17}
\int_{X^k} D(\omega_{B}^{x}\| \omega_{B})\, dx = D(\omega_{XB}\| \mu_{X^{k}}  \otimes \omega_B)\ ,
\end{align}
where $\mu_{X^{k}} $ denotes the restriction of the Lebesgue measure on $X^k$. Note that $\mu_{X^{k}} $ is now a positive finite functional such that we can apply the approximation result for the quantum relative entropy of states on a von Neumann algebra along a net of increasing subalgebras in $L^\infty(X^k)$ (Lemma~\ref{lem:vNconv}). In particular, we take the net of subalgebras given by the step functions corresponding to the fixed partitions $\cP_\alpha^k$ obtained by restricting $\cP_\alpha$ to $X^k$. 
We assume here that the covering $\{X_k\}$ is taken such that it is compatible with the partitions $\cP_\alpha$ for a small enough $\alpha$ such that $\cP_\alpha^k$ is balanced as well. Note that such a covering exists since one can for instance take the the sets of a fixed partition $\cP_\alpha$ for a large enough $\alpha$.
Let us denote the corresponding alphabet of the induced discrete and finite abelian algebra by $X_\alpha^k$. Hence, we obtain that 
\begin{align}\label{eq1:pf:prop:vNapprox}
 - h(X|B)_{\omega} = \lim_{n\rightarrow \infty}  \lim_{\alpha\rightarrow 0} \sum_{k=1}^n D(\omega_{X_\alpha^k B}\|\mu_{X_\alpha^k}\otimes\omega_{B})\ . 
\end{align}
where $\omega_{X_\alpha^kB}$ and $\mu_{X_\alpha^k}\otimes\omega_{B}$ are states in $\ell^\infty(X_\alpha^k)\otimes\cM_{B}$ and defined as in~\eqref{eq:DiscState}. We therefore have that $\mu_{X_\alpha^k} = \alpha \idty$, where the identity is the one in $\ell^\infty(X_\alpha^k)$, and it follows by an elementary property of the quantum relative entropy (Lemma~\ref{lem:petz2}) that  
\begin{align}\label{eq2:pf:prop:vNapprox}
D(\omega_{X_\alpha^k B}\|\mu_{X_\alpha^k}\otimes\omega_{B})=D(\omega_{X_\alpha^k B}\|\idty\otimes\omega_{B})-p_k\log\alpha\ ,
\end{align}
where $p_k = \int_{X^k} \omega_{B}^x(\idty) dx$ is the probability that an event in the interval $X^k$ occurs.  Hence, in order to obtain the approximation result in the proposition, we have to show that the limits on the right hand side of~\eqref{eq1:pf:prop:vNapprox} can be interchanged. For that, it is sufficient to show that the sum $\sum f_k(\alpha)$ with $f_k(\alpha)=D(\omega_{X_\alpha^k B}\|\mu_{X_\alpha^k}\otimes\omega_{B})$ converges uniformly. By assumption, $H(X_{\cP_{\alpha_0}}|B)_{\omega}<\infty$, and due to the monotonicity of the quantum relative entropy under restrictions (Lemma~\ref{lem:petz1}), we then get that $f_k(\alpha_0)\leq f_k(\alpha) \leq f_k(0) $ for all $k$. Together with~\eqref{eq2:pf:prop:vNapprox} it follows that
\begin{align}
h(X|B)_{\omega} = - \sum_k f_k(0) \leq -\sum_k f_k(\alpha_0)  =  H(X_{\cP_{\alpha_0}} |B)_{\omega}+\log\alpha_0 <\infty\ ,
\end{align}
and since by assumption $h(X|B)_{\omega}>-\infty$, we conclude that $h(X|B)_{\omega}$ is finite. Further, we have that $|f_k(\alpha)| \leq |f_k(\alpha_0)| + |f_k(0)|=M_k $. Note that the terms $f_k(\alpha)$ in the sum can be negative or positive and we need lower and upper bounds in order to bound the absolute value of $f_k(\alpha)$. Using the Weierstrass uniform convergence criteria, it remains to show that $\sum_k M_k$ is finite. The series $\sum_k |f_k(0)|$ is finite since it is upper bounded by $\int |D(\omega_{B}^{x}\| \omega_{B})|\, dx$, which is finite since $h(X|B)_{\omega}$ is finite. Using~\eqref{eq2:pf:prop:vNapprox} and the fact that $D(\omega_{X_\alpha^k B}\|\idty\otimes\omega_{B})\geq 0$ for all $k$, it is easy to see that the series $\sum_k |f_k(\alpha_0)| $ is bounded by $H(X_{\cP_{\alpha_0}}|B)+\log(\alpha_0)$ (which is finite by assumption). This concludes the first statement of the proposition.

The second statement follows from the first together with the chain rule for the quantum relative entropy (Lemma~\ref{lem:petz3}).
\end{proof}

\subsection{Quantum Conditional Min- and Max-Entropy}\label{sec:MinMax}

Quantum conditional min- and max-entropy have already been investigated on infinite-dimensional Hilbert spaces~\cite{Furrer10} and von Neumann algebras~\cite{Berta11,FurrerPhD}. For finite classical systems $X$, the conditional min- and max-entropy for a state $\omega_{XB}$ on the bipartite system $\cM_{XB}=L^\infty(X)\otimes \cM_B$ with $\cM_B$ a von Neumann algebra are given by~\cite{Berta11}
\begin{align}
&H_{\min}(X|B)_{\omega} =-\log\sup\left\{\sum_{x}\omega_{B}^{x}(E_{B}^{x}): \; E_{B}^{x}\in\cM_{B},E_{B}^{x}\geq0,\sum_{x}E_{B}^{x}=\idty_{B}\right\} \label{eq:Guessing}\\
&H_{\max}(X|B)_{\omega}=2 \log\sup\left\{\sum_{x}\sqrt{F(\omega_{B}^{x},\sigma_{B})} :\; \sigma_{B}\in\cS(\cM_{B})\right\}\ , \label{eq:SecKey}
\end{align}
where $F(\cdot,\cdot)$ denotes the fidelity~\eqref{eq:fidelity}. These quantities admit natural extensions to classical-quantum systems where the classical variable takes values in an arbitrary $\sigma$-finite measure space.

\begin{definition}\label{def:cont_entropy}
Let $\cM_{XB}= L^{\infty}(X)\otimes\cM_{B}$ with $(X,\Sigma,\mu)$ a $\sigma$-finite measure space, $\cM_{B}$ a von Neumann algebra, and $\omega_{XB}\in\cS_\leq(\cM_{XB})$. Then, the conditional min-entropy of $X$ given $B$ is defined as
\begin{align}\label{eq:mincont}
\dmin{X}{B}_{\omega} =-\log\sup \left\{\int\omega_{B}^{x}(E_{B}^{x})\,d\mu(x):\;E\in L^\infty(X) \otimes \cM_{B},E\geq0, \int E_{B}^{x}\,d\mu(x)\leq\idty_{B}\right\}\ .
\end{align}
Furthermore, the conditional max-entropy of $X$ given $B$ is defined as
\begin{align}\label{eq:max}
\dmax{X}{B}_{\omega}= 2 \log\sup\left\{\int \sqrt{F(\omega_{B}^{x},\sigma_{B})}\,d\mu(x):\;\sigma_{B}\in\cS(\cM_{B})\right\}\ .
\end{align}
\end{definition}

The quantities in~\eqref{eq:mincont} and~\eqref{eq:max} are well defined since the integrands are measurable and positive. An important example is given for $X=\mathbb R$. Then, for trivial quantum memory~$\cM_{B}=\mathbb C$, the differential min- and max-entropy correspond to the differential R\'enyi entropy of order $\infty$ and $1/2$, respectively,
\begin{align}
&h_{\min}(X)_{\omega}=-\log\Vert \omega_X\Vert_{\infty}\label{eq:renyiinfty}\\
&h_{\max}(X)_{\omega}= 2\log\int\sqrt{\omega^x} \, dx = \log\Vert\omega_{X}\Vert_{\frac{1}{2}}\label{eq:renyi12}\ ,
\end{align}
where $\Vert \cdot \Vert_{p}$ denotes the usual p-norm on $L^p(\mathbb R)$. We note that like any differential entropy, the differential conditional min- and max-entropy can get negative. In particular,
\begin{align}
-\infty\leq h_{\min}(X)_{\omega}<\infty,\quad-\infty<h_{\max}(X)_{\omega}\leq\infty\ .
\end{align}
and the same bounds also hold for the conditional versions in~\eqref{eq:mincont} and~\eqref{eq:max}. 

In the case where the measure space $X$ is discrete and equipped with the counting measure, we retrieve the usual definitions as in~\eqref{eq:Guessing} and~\eqref{eq:SecKey}, with sums now involving infinitely many terms. We therefore use uppercase letters for the entropies, $H_{\min}(X|B)_{\omega}$ and $H_{\max}(X|B)_{\omega}$,  if $X$ is discrete. For $X$ having infinite cardinality, we can assume that $X\cong \mathbb N$ and since all the terms inside the sums are positive, the conditional min- and max-entropy can be obtained from finite sum approximations
\begin{align}
&H_{\min}(X|B)_{\omega}=-\log\sup_{m}\sup\left\{\sum_{x=1}^{m}\omega_{B}^{x}(E_{B}^{x}): \; E_{B}^{x}\in\cM_{B},E_{B}^{x}\geq0,\sum_{x=1}^{m}E_{B}^{x}\leq\idty_{B}\right\}\label{eq:minapprox}\\
&H_{\max}(X|B)_{\omega}=2 \log\sup_{m}\sup\left\{\sum_{x=1}^{m}\sqrt{F(\omega_{B}^{x},\sigma_{B})} :\; \sigma_{B}\in\cS(\cM_{B})\right\}\label{eq:maxapprox}\ .
\end{align}

The following proposition shows that the differential conditional min- and max-entropy for coarse grained measure spaces are retrieved from the discrete quantities in the regularized limit of finer and finer discretization. 

\begin{proposition}\label{thm:MinMaxApprox}
Let $\cM_{XB}=L^{\infty}(X)\otimes\cM_{B}$ with $\cM_{B}$ a von Neumann algebra and $(X,\Sigma,\mu,\{\cP_\alpha\})$ a coarse grained measure space. Then, we have that for $\omega_{XB}\in\cS(\cM_{XB})$, 
\begin{align}\label{thm,eq2:MinMaxApprox}
\dmin XB_\omega=\lim_{\alpha\rightarrow 0}\Big(\Hmin{X_{\cP_\alpha}}B_\omega+\log\alpha\Big)\ ,
\end{align}
where $\omega_{X_{\cP_{\alpha}}B}$ is defined as in~\eqref{eq:DiscState}. Furthermore, if there exists an $\alpha_{0}>0$ such that $H_{\max}(X_{\cP_{\alpha_{0}}})_\omega<\infty$, then we have that
\begin{align}\label{thm,eq1:MinMaxApprox}
\dmax XB_\omega=\lim_{\alpha\rightarrow 0}\Big(\Hmax {X_{\cP_\alpha}}B_\omega+\log\alpha\Big)\ .
\end{align}
\end{proposition}

A similar result under additional conditions and with different techniques is derived in the thesis of one of the authors~\cite{FurrerPhD}. 
We emphasize that the limits for $\alpha \rightarrow 0$ in~\eqref{thm,eq2:MinMaxApprox} and~\eqref{thm,eq1:MinMaxApprox} can be replaced by an infimum over $\alpha$. This emerges directly from the following proof of Proposition~\ref{thm:MinMaxApprox}.

\begin{proof}
We start with the differential conditional min-entropy. Let us fix an $\alpha_0$ and consider $\cP_{\alpha_0}=\{I^{\alpha_0}_l\}_{l\in\Lambda}$ where we can assume that $\Lambda = \{1,2,3,...\} \subset \mathbb N$. For $k\in\Lambda$, we then define $C_k=\bigcup_{l=1}^k I_l^{\alpha_0}$ which is compact according to the definition of a coarse graining. We then write the differential min-entropy as 
\begin{align}
h_{\min}(X|B)_{\omega} =-\log\sup_{k}\sup\Big\{\omega_{XB}(E):\;E\in L^\infty(C_k) \otimes \cM_{B},E\geq0, \int E_{B}^{x}\,d\mu(x)\leq\idty_{B}\Big\}\ .
\end{align}
Since $C_k$ is compact, the set of step functions $\cT^k=\bigcup_{\alpha\leq \alpha_0} \cT^k_\alpha$ with $\cT^k_\alpha$ the step functions corresponding to partitions $\cP^k_{\alpha}$ defined as the restriction of $\cP_\alpha$ onto $C_k$ is $\sigma$-weakly dense in $L^\infty(C_k)$. Because $\omega_{XB}$ is $\sigma$-weakly continuous we get that
\begin{align}\label{pf,thm:MinMaxApprox,eq1}
h_{\min}(X|B)_{\omega} =-\log\sup_{k}\sup_{\alpha} \Big\{\omega_{XB}(E):E\in \cT^k_\alpha \otimes \cM_{B},E\geq0, \int E_{B}^{x}\,d\mu(x)\leq\idty_{B}\Big\}\ ,
\end{align}
where we used that a $\{\cP^k_\alpha\}$ is an ordered family of partitions. By exchanging the two suprema, we find that the right hand side of~\eqref{pf,thm:MinMaxApprox,eq1} reduces to the infimum of $H_{\min}(X_{\cP_{\alpha}}|B)_{\omega}+\log\alpha$ over $\alpha$, with $\omega_{X_{\cP_{\alpha}B}}$ defined as in~\eqref{eq:DiscState}. Finally, we note that since the expression on the right hand side of~\eqref{pf,thm:MinMaxApprox,eq1} is monotonic in $\alpha$, the infimum over $\alpha$ can be exchanged by the limit $\alpha \rightarrow 0$. 

To show the approximation of the differential conditional max-entropy, we start by using a different but equivalent expression for the fidelity~\cite{Alberti-1983}
\begin{align}
h_{\max}(X|B)_{\omega}=2\log\sup\left\{\int \sup_{U(x)\in\pi(\cM_{B})'}|\bra{\xi_{\omega}^{x}}U(x)\ket{\xi_{\sigma}}|dx : \; \sigma_{B}\in\cS(\cM_{B})\right\}\ ,
\end{align}
where $\pi$ is some fixed representation of $\cM_{B}$ in which all $\omega_{B}^{x}$ and $\sigma_{B}$ admit vector states $\ket{\xi^{x}_{\omega}}$, $\ket{\xi_{\sigma}}$, respectively, and the supremum is taken over all elements $U(x)\in\pi(\cM_{B})'$ with $\|U(x)\|\leq1$. We note that we can always choose $U(x)$ such that $\bra{\xi_{\omega}^{x}}U(x)\ket{\xi_{\sigma}}$ is positive. It follows by the $\sigma$-finiteness of the measure space together with the theorem of monotone convergence, that we can find a sequence of sets $X^n$ all having finite measure, and
\begin{align}
h_{\max}(X|B)_{\omega}&=2\log\sup_{\sigma_{B}\in\cS(\cM_{B})}\lim_{n\to\infty}\sup_{U(x)\in\pi(\cM_{B})'}\int_{X^n}\bra{\xi_{\omega}^{x}}U(x)\ket{\xi_{\sigma}}dx\ .
\end{align}
For later reasons we note that the sequence $X^n$ can be chosen such that it is compatible with the partitions in the sense that for every $n$ the restriction of $\cP_\alpha$ onto $X^n$ forms again a balanced partition with measure $\alpha$. One can take for instance $X^n$ to be generated by finite increasing unions of the sets in a partition $\cP_{\alpha_0}$ for a fixed $\alpha_0$. Then for all $\alpha \leq \alpha_0$ the condition is satisfied. It then follows from disintegration theory of von Neumann algebras, Ref.~\onlinecite[Chapter IV.7]{Takesaki1}, that the expression involving the second supremum and the integral can again be recognized as a fidelity, more precisely, as the square root of $F(\omega_{X^n B}, \mu_{X^n} \otimes\sigma_B)$, where $\omega_{X^n B}$ is the state restricted to the subalgebra $L^\infty(X^n) \otimes\cM_B \subseteq L^\infty(X) \otimes\cM_B$, and $\mu_{X^n}$ is the Lebesgue measure restricted to the set $X^n$. Because $X_n$ is of finite measure, $\mu_{X^n}$ can be identified as a positive functional on $L^\infty(X^n) \otimes\cM_B$. The fidelity between the two positive forms $\omega_{X^n B}$ and $\mu_{X^n} \otimes\sigma_B$ can be approximated by evaluating~\cite{Alberti:2000ve}
\begin{align}
F(\omega_{X^n B}, \mu_{X^n} \otimes\sigma_B) = \inf \left\{ \sum_j \omega_{X^n B}(e_j)\, \mu_{X^n} \otimes\sigma_B(e_j) \right\}\ ,
\end{align}
where $\{e_j\} \subset L^\infty(X^n) \otimes\cM_B$ are finitely many orthogonal projections summing up to the identity. Using the same reasoning as for the conditional min-entropy, we can restrict this infimum to finite sets of orthogonal projections in $\cT \otimes \cM_B$. Since any such projection is of the form $\chi_{I^\alpha_{k}}\otimes P^{k}_{B}$, for a projection $P^k_B \in \cM_B$, we find that
\begin{align}
F(\omega_{X^n B},\mu_{X^n}\otimes\sigma_B)&=\inf_{\alpha>0}F(\omega_{X^{n}_{\cP_{\alpha}}B},\mu_{X^{n}_{\cP_{\alpha}}}\otimes\sigma_{B})=\lim_{\alpha\to0}F(\omega_{X^{n}_{\cP_{\alpha}}B},\mu_{X^{n}_{\cP_{\alpha}}}\otimes\sigma_{B})\ .
\end{align}
Note that the infimum can be replaced by the limit since the family $\{P_\alpha\}$ is ordered and the fidelity is monotonic under restrictions~\cite{Alberti-1983}. This leads to
\begin{align}\label{eq:limits_wrong}
h_{\max}(X|B)_{\omega} &= 2 \log\sup_{\sigma_{B}\in\cS(\cM_{B})}\lim_{n\to\infty}\lim_{\alpha\to0}\sqrt{F(\omega_{X^n_{\cP_{\alpha}}B}, \mu_{X^n_{\cP_{\alpha}}} \otimes\sigma_B)}\ ,
\end{align}
and in order to proceed, we have to interchange the limits. For this, we define
\begin{align}
f_{n}(\sigma,\alpha)=\sqrt{F(\omega_{X^n_{\cP_{\alpha}}B},\mu_{X^n_{\cP_{\alpha}}}\otimes\sigma_B)}=\sum_{k\in\Lambda(\alpha,n)}\sqrt{\alpha\cdot F(\omega_{B}^{\cP_{\alpha}^{n},k},\sigma_B)}\ ,
\end{align}
where we have used that $\mu_{X^n}$ restricted to $\cN_{\alpha}$ is just the counting measure multiplied by $\alpha$. Since $f_{n}(\sigma,\alpha)$ is monotonously increasing in $\alpha$, there exists by assumption an $\alpha_{0}$ such that
\begin{align}
f_{n}(\sigma,\alpha)\leq\sum_{k\in\Lambda(\alpha_{0},n)}\sqrt{\alpha_{0}\cdot F(\omega_{B}^{\cP_{\alpha_{0}}^{n},k},\sigma_B)}\leq\sum_{k\in\Lambda(\alpha_{0},n)}\sqrt{\alpha_{0}\cdot\omega_{B}^{\cP_{\alpha_{0}}^{n},k}(\idty)}
\end{align}
is finite in the limit $n\to\infty$. It follows by the Weierstrass' uniform convergence theorem that the sequence $f_{n}(\sigma,\alpha)$ converges uniformly in $\sigma$ and $\alpha$ to the limiting function $f(\sigma,\alpha)=\lim_{n\to\infty}f_{n}(\sigma,\alpha)$. Hence, the limits in~\eqref{eq:limits_wrong} can be interchanged, and we arrive at
\begin{align}
h_{\max}(X|B)_{\omega}=2\log\sup_{\sigma_{B}\in\cS(\cM_{B})}\lim_{\alpha\to0}f(\sigma,\alpha)=2\log\sup_{\sigma_{B}\in\cS(\cM_{B})}\inf_{\alpha>0}f(\sigma,\alpha)\ .
\end{align}
In the last step, we need to interchange the supremum with the infimum. For this, we extend the function $f(\sigma,\alpha)$ by setting $f(\sigma,0)=\int\sqrt{F(\omega_{B}^{x},\sigma_{B})}d\mu(x)$ such that we can write
\begin{align}\label{eq:minimax_wrong}
h_{\max}(X|B)_{\omega}=2\log\sup_{\sigma_{B}\in\cS(\cM_{B})}\inf_{\alpha\in[0,\alpha_{0}]}f(\sigma,\alpha)\ .
\end{align}
Since $f_{n}(\sigma,\alpha)$ converges uniformly in $\sigma$ and $\alpha$, and $f(\sigma,\alpha)$ is monotonically increasing in $\alpha$, we have
\begin{align}
\inf_{\alpha>0}f(\sigma,\alpha)=f(\sigma,0)\ .
\end{align}
Hence, we find that
\begin{align}
\sup_{\sigma_{B}\in\cS(\cM_{B})}\inf_{\alpha\in[0,\alpha_{0}]}f(\sigma,\alpha)=\inf_{\alpha\in[0,\alpha_{0}]}\sup_{\sigma_{B}\in\cS(\cM_{B})}f(\sigma,\alpha)\ ,
\end{align}
and by using that $f(\sigma,\alpha)$ is monotonically increasing in $\alpha$ we obtain with~\eqref{eq:minimax_wrong} that
\begin{align}
h_{\max}(X|B)_{\omega}=\lim_{\alpha\rightarrow0}\Big(H_{\max}(X_{\cP_{\alpha}}|B)_{\omega}+\log\alpha\Big)\label{eq:maxuncert}\ .
\end{align}
\end{proof}

We note that the condition $H_{\max}(X_{\cP_{\alpha_{0}}})_\omega<\infty$ required for the approximation of the differential conditional max-entropy does not follow from $h_{\max}(X)_\omega<\infty$. Indeed, there exist states $\omega_{X}\in L^{1}(X)$ with $h_{\max}(X)_\omega<\infty$ but $H_{\max}(X_{\cP_{\alpha}})_\omega=\infty$ for all $\alpha>0$. As an example of such a state $\omega_X$ for $X=\mathbb R$ take the normalization of the function which is equal to $1$ for $x\in [k,k-1/k^2]$, $k\in\mathbb N$, and $0$ else. But conversely, $H_{\max}(X_{\cP_{\alpha}})_\omega < \infty$ implies that $h_{\max}(X)_\omega<\infty$ since the relation $h_{\max}(X|B)_\omega\leq H _{\max}(X_{\cP_{\alpha}}|B)_\omega+\log\alpha$ holds for all $\alpha>0$ and $\omega_{XB}\in\cS(\cM_{XB})$.

In the important case where $X=\mathbb R$, the condition $H_{\max}(X_{\cP_{\alpha_{0}}})_\omega<\infty$ is satisfied under the assumption that the second moment of the distribution $\omega_X$ is finite (which is often a valid assumption in physical applications).

\begin{lemma}\label{lem:MaxCond}
Let $X =\mathbb{R}$ and $\omega\in \cS(L^\infty(X))$ such that $\int \omega(x) x^2 d x<\infty$. Then, there exists a partition $\cP_\alpha$ of $X$ into intervals of length $\alpha >0$ such that $H_{\max}(X_{\cP_\alpha})_\omega<\infty$.
\end{lemma} 

\begin{proof}
Let us fix $\alpha$ and take the partition $\cP_\alpha$ of $X$ into intervals $I_k=[k\alpha, (k+1)\alpha] $ for $k\in \mathbb Z$. The max-entropy $H_{\max}(X_{\cP_\alpha})_\omega$ is finite if and only if $\sum_{k} \sqrt{\omega_k} $ with $\omega_k=\int_{I_k} \omega(x)  dx$ is finite. By means of the monotone convergence theorem, we can write 
\begin{align}
\int x^2 \omega(x) \, dx =  \sum_{k\geq 0} \int_{I_k} \omega(x) x^2 \, dx + \sum_{k < 0} \int_{I_k} \omega(x) x^2 \, dx\ .
\end{align}
For the following we only consider the sum over $k\geq 0$, but the same argument can also be applied to the sum over $k < 0$. From the monotonicity of the square and the definition of $I_k$ follows that $\int_{I_k} \omega(x) x^2  dx \geq (\alpha k)^2 \int_{I_k} \omega(x) dx =  (\alpha k)^2 \omega_k$. Hence, we find that 
\begin{align}
 \alpha^2  \sum_{k}  k^2\omega_k  \leq  \sum_{k} \int_{I_k} \omega(x) x^2 \, dx < \infty\ ,
\end{align} 
and since all terms are positive we conclude that the sum $\sum  k^2\omega_k$ converges absolutely. We set $\Delta = \{k \in\mathbb N \, : \,  k^2\sqrt{\omega_k} \geq 1 \}$ and write
\begin{align}
\sum_{k\in \mathbb N}  k^2\omega_k  = \sum_{k\in \Delta}  k^2\omega_k  + \sum_{k\in \NN \backslash \Delta}  k^2\omega_k  \geq  \sum_{k\in \Delta}  \sqrt{\omega_k}   \ ,
\end{align}
where we used that absolute converging series can be reordered and that $\sum  k^2\omega_k= \sqrt{\omega_k}(k^2\sqrt{\omega_k}) \geq \sqrt{\omega_k}$ for all $k\in\Delta$. Hence, we find that  $ \sum_{k\in \Delta}  \sqrt{\omega_k}$ converges absolutely. Moreover, by definition of $\Delta$, it holds that $\sqrt{\omega_k} < 1/k^2$ for all $k\in \NN \backslash \Delta$ such that $\sum_{k\in \NN\backslash \Delta} \sqrt{\omega_k} \leq \sum_{k\in \NN\backslash \Delta} {1}/{k^2} < \infty$. 
Using again that absolutely converging series can be reordered, we finally obtain 
\begin{align}
\sum_{k\in \mathbb N} \sqrt{\omega_k} = \sum_{k\in \Delta} \sqrt{\omega_k } +  \sum_{k\in \NN\backslash \Delta} \sqrt{\omega_k}   < \infty  \ .
\end{align}
\end{proof}



\section{Entropic Uncertainty Relations in the Presence of Quantum Memory}\label{sec:relations}

In the following sections, we first derive the uncertainty relations for arbitrary measurements and then discuss the special case of position and momentum observables. Our starting point is a recent proof technique developed by Coles et al.~\cite{PhysRevLett.108.210405} (see also Ref.~\onlinecite{Tomamichel11}) for finite-dimensional Hilbert spaces, which we generalize to von Neumann algebras and continuous measure spaces by means of the approximation results derived in Section~\ref{sec:entropy}. The advantage of the applied proof strategy is that it only relies on basic properties of the involved entropies. We start with the uncertainty relation for the quantum conditional min- and max-entropy. 


\subsection{Uncertainty Relations in Terms of Quantum Conditional Min- and Max-Entropy}\label{sec:minmaxrelations}

An uncertainty relation for conditional min- and max-entropy was first derived in the finite-dimensional setting~\cite{Tomamichel11}, and then generalized to measurements with a finite number of outcomes on von Neumann algebras~\cite{Berta11}. Before stating the extension to measurements with continuous outcomes, we first prove the uncertainty relation for the case of measurements with an infinite, but countable number of outcomes. We emphasize here that this result cannot directly be inferred from the similar relation for a finite measurement range since the uncertainty relation does not generalize to non-normalized POVMs.  

\begin{proposition}\label{prop:minmaxdisc}
Let $\cM_{ABC}$ be a tripartite von Neumann algebra, $\omega_{ABC}\in\cS(\cM_{ABC})$, $X$ and $Y$ countable, and $E_{X}=\{E_x\}_{x\in X} \in\mathrm{Meas}(X,\cM_{A})$ and $F_{Y}=\{F_y\}_{y \in Y} \in\mathrm{Meas}(Y,\cM_{A})$. Then, we have that
\begin{align}
H_{\max}(X|B)_{\omega}+H_{\min}(Y|C)_{\omega}\geq -\log c(E_{X},F_{Y})\ ,
\end{align}
where the overlap of the measurements is given by
\begin{align}\label{eq:ccountable}
c(E_{X},F_{Y})= \sup_{x,y}\|E_x^{1/2} F_y^{1/2}\|^{2}\ .
\end{align}
\end{proposition}

\begin{proof}
The main difference to the proofs of the uncertainty relations in~\cite{Tomamichel11,PhysRevLett.108.210405} is that we have to take an infinite number of outcomes into account. We achieve this by first showing an inequality for sub-normalized measurements with a finite number of outcomes, and then use a limit argument to obtain the uncertainty relation for measurements with an infinite number of outcomes. We describe sub-normalized measurements $E_X$ and $F_Y$ by a finite collection of positive operators $\{E_x\}_{x\in X}$ and $F_Y=\{F_y\}_{y\in Y}$, which sum up to $M=\sum_{x}E_{x}\leq\idty$ and $\sum_{y}F_{y}\leq\idty$.

Let us take a Hilbert space $\cH$ where $\cM_{ABC}\subseteq\cB(\cH)$ is faithfully embedded and there exists a purifying vector $\ket{\psi}\in\cH$ for $\omega_{ABC}$, that is, $\omega_{ABC}(\cdot)=\bra{\psi}\cdot\psi\rangle$. We choose a Stinespring dilation (see Ref.~\onlinecite[Theorem 4.1]{Paulsen}) for $E_{X}$ of the form
\begin{align}
V:\cH\rightarrow\cH\otimes\mathbb{C}^{X}\otimes\mathbb{C}^{X'},\quad V\ket{\psi}=\sum_{x}E_{x}^{1/2}\ket{\psi}\otimes\ket{x,x}\ ,
\end{align}
where $\mathbb{C}^{X}$ denotes a quantum system with dimension $X$ in which the classical output of the measurement $E_{X}$ is embedded, and $X\cong X'$. 
Since $\ket{\psi}\in\cH$ is a purifying vector of $\omega_{ABC}$, we have that $V\ket{\psi}\in\cH\otimes\mathbb{C}^{X}\otimes\mathbb{C}^{X'},$ is a purifying vector of $\omega_{XB}=\omega_{AB}\circ E_X$. Denoting the commutant of $\cM_{ABC}$ in $\cB(\cH)$ by $\cM_{D}$, we find that the purifying system of $\omega_{XB}$ is $\cB(\mathbb{C}^{X'})\otimes\cM_{ACD}$. It then follows from the duality of the conditional min- and max-entropy (Lemma~\ref{lem:duality}) that
\begin{align}\label{eq:first}
H_{\max}(X|B)_{\omega}=-H_{\min}(X|X'ACD)_{\psi\circ V}\ ,
\end{align}
where $\psi_{ACD}\circ V(\cdot)=\bra{\psi}V^{*}(\cdot)V\psi\rangle$. Since the conditional min-entropy can be written as a max-relative entropy (Definition~\ref{lem:HminDmax}), we have that
\begin{align}
-H_{\min}(X|X'ACD)_{\psi\circ V}=\inf_{\sigma_{X'ACD}}D_{\max}(\psi_{ACD}\circ V \|\tau_{X}\otimes\sigma_{X'ACD})\ ,
\end{align}
where the max-relative entropy is given by (Definition~\ref{def:maxrelative})
\begin{align}
D_{\max}(\psi_{ACD}\circ V\|\tau_{X}\otimes\sigma_{X'ACD})=\inf\{\iota\in\mathbb{R}:\psi_{ACD}\circ V\leq2^{\iota}\cdot\tau_{X}\otimes\sigma_{X'ACD}\}\ ,
\end{align}
the infimum is over $\sigma_{X'ACD}\in\cS(\cM_{X'ACD})$, and $\tau_{X}$ denotes the trace on $\cB(\mathbb{C}^{X})$. Let us now define the completely positive map $\cE:\cB(\cH) \rightarrow \cB(\cH\otimes\mathbb{C}^{X}\otimes\mathbb{C}^{X'})$ given by $\cE(a)= V a V^*$. The map is sub-unital since $\cE(\idty) = VV^*$ and $\Vert V V^* \Vert = \Vert V^* V\Vert = \Vert\sum_{x}E_x \Vert = \Vert M \Vert \leq 1$. Due to the monotonicity of the max-relative entropy under applications of sub-unital, completely positive maps (Lemma~\ref{lem:mon_Dmax}), we obtain for fixed $\sigma_{X'ACD}\in\cS(\cM_{X'ACD})$,
\begin{align}\label{eq:uncertdev1}
D_{\max}(\psi_{ACD}\circ V\|\tau_{X}\otimes\sigma_{X'ACD})&\geq D_{\max}((\psi_{ACD}\circ V)\circ\cE \|(\tau_{X}\otimes\sigma_{X'ACD})\circ\cE )\\
&=D_{\max}(\omega^V_{ACD}\|\gamma^{\sigma,V}_{ACD})\ ,
\end{align}
where we denoted $\omega^V_{ACD}= (\psi_{ACD}\circ V) \circ\cE $ and $\gamma^{\sigma,V}_{ACD}=(\tau_{X}\otimes\sigma_{X'ACD})\circ\cE $. Due to the fact that $V^*V= M$, we have that $\omega^V_{ACD}(\cdot) = \psi_{ACD}\circ V\circ V^{*}(\cdot)=\Scp{\psi}{V^{*}V(\cdot)V^{*}V\psi}=\omega_{ACD}(M \cdot M)$ with $\omega_{ACD}$ the state on $\cM_{ACD}$ corresponding to $\ket{\psi}$. Using once more the monotonicity of the max-relative entropy under application of channels (Lemma~\ref{lem:mon_Dmax}), we obtain by first restricting onto the subalgebra $\cM_{AC}$ and then measuring the $A$ system with $F$,
\begin{align}\label{pf,thm:minmaxcont,eq2}
D_{\max}(\omega^V_{ACD}\|\gamma^{\sigma,V}_{ACD})\geq D_{\max}(\omega^V_{AC}\|\gamma^{\sigma,V}_{AC})\geq D_{\max}(\omega^V_{YC}\|\gamma^{\sigma,V}_{YC})\ ,
\end{align}
where we set $\omega^V_{YC}=\omega^V_{AC}\circ F_Y$ and $\gamma^{\sigma,V}_{YC} = \gamma^{\sigma,V}_{AC} \circ F_Y$.  By definition, we have that $\gamma^{\sigma,V}_{YC} (a) = \sum_y \gamma^{\sigma,V}_{AC}(F_y a_y)$ for $a=(a_y)\in \cM_{YC}$. Hence, it holds for all positive $a_y\in\cM_C$ with $y\in Y$ that
\begin{align}\label{eq:boundC}
\gamma^{\sigma,V}_{AC}(F_y a_y) = \tau_X\otimes \sigma_{X'AC}(VF_y  a_y V^*) = \sum_{x}\sigma^{x,x}_{AC}( \sqrt{E_x}F_y\sqrt{E_x} a_y ) \leq 
\sup_{x,y}\left\|E_{x}^{1/2} F_{y}^{1/2} \right \|^{2} \sigma_C(a_y)\ ,
\end{align} 
where we denoted $\sigma_{X'ACD}=(\sigma^{x,x'}_{ACD})$, and used that $a_y$ commutes with $E_x$ and $F_y$. Thus, we conclude that
\begin{align}
\gamma^\sigma_{YC}\leq\sup_{x,y}\left\|E_{x}^{1/2} F_{y}^{1/2}\right\|^{2}\cdot\tau_{Y}\otimes\sigma_{C}\ .
\end{align} 
By using elementary properties of the max-relative entropy (Lemma~\ref{lem:minmax_elementary1} and Lemma~\ref{lem:minmax_elementary2}), it then follows for any $\sigma_{X'ACD}\in\cS(\cM_{X'ACD})$ that
\begin{align}
D_{\max}(\omega^V_{YC}\|\gamma^\sigma_{YC})&\geq D_{\max}(\omega^V_{YC}\|\tau_{Y}\otimes\sigma_{C})-\log\sup_{x,y}\left\|E_{x}^{1/2} F_{y}^{1/2}\right\|^{2}\\
&\geq\inf_{\eta}D_{\max}(\omega^V_{YC}\|\tau_{Y}\otimes\eta_{C})-\log\sup_{x,y}\left\|E_{x}^{1/2} F_{y}^{1/2}\right\|^{2}\\
&=-H_{\min}(Y|C)_{\omega^V}-\log\sup_{x,y}\left\|E_{x}^{1/2} F_{y}^{1/2}\right\|^{2}\label{eq:uncertdev2}\ ,
\end{align}
where the infimum is over $\eta_{C}\in\cS(\cM_{C})$, and we used again that the conditional min-entropy can be written as a max-relative entropy (Definition~\ref{lem:HminDmax}). Combining this with all the steps going back to~\eqref{eq:first}, we therefore obtain 
\begin{align} \label{eq:URdiscSubnormal}
H_{\max}(X|B)_{\omega} \geq -H_{\min}(Y|C)_{\omega^V}-\log\sup_{x,y}\left\|E_{x}^{1/2} F_{y}^{1/2}\right\|^{2}\ .
\end{align}
Recall that $\omega^V_{YC} = (\omega^{V,y}_{C})_y$ with $\omega^{V,y}_{C}(\cdot)= \omega(M F_y M \cdot)$, and thus, if $E$ is normalized we obtain the uncertainty relation for measurements with a finite number of outcomes.

Let us now lift the relation to the case of discrete $X$ and $Y$ with infinite cardinality. We take sequences of increasing finite subsets $X_1 \subset X_2 \subset ... \subset X$  and $Y_1 \subset Y_2 \subset ...\subset  Y$ such that $\bigcup_n X_n = X$ and $\bigcup_n Y_n = Y$. The strategy is to apply the inequality~\eqref{eq:URdiscSubnormal} derived for sub-normalized measurements to $E_{X_n}=\{E_x\}_{x\in X_n}$ and $F_{Y_m}=\{F_y\}_{y\in Y_m}$. It is straightforward to see that~\eqref{eq:URdiscSubnormal} for fixed $n$ and $m$ reads as
\begin{align} 
H_{\max}(X_n|B)_{\omega} \geq -H_{\min}(Y_m|C)_{\omega^n}-\log\sup_{x \in X ,y\in Y}\left\|E_{x}^{1/2} F_{y}^{1/2}\right\|^{2}\ ,
\end{align}
where we denoted $\omega_{X_nB} = \omega_{AB} \circ E_{X_n}$ and $\omega^n_{Y_mC} = (\omega^{n,y}_{C})_{y\in Y_m}$ with $\omega^{n,y}_C(\cdot)=\omega_{AC}({M_n} F_y {M_n} \cdot)$, and $M_n=\sum_{x\in X_n} E_x$. Note that we already used that taking the supremum over $X$ and $Y$ instead of $ X_n$ and $y\in Y_m$ only decreases the right hand side of the above inequality. 
We now take the limit for $n\rightarrow \infty$ on both sides. By using the definition of the conditional max-entropy in~\eqref{eq:maxapprox}, it is straightforward to see that $ H_{\max}(X_n|B)_{\omega}$ converges to $H_{\max}(X|B)_{\omega}$ for $n\rightarrow \infty$. The only term on the right hand side depending on $n$ is the conditional min-entropy of the state $\omega^n_{Y_mC}$, which is given by (see~\eqref{eq:Guessing})
\begin{align}
H_{\min}(Y_m|C)_{\omega^n} = -\log\sup_{G} \sum_{y\in Y_m}\omega_{AC}(M_{n}F_{y}M_{n}G_{y})\ ,
\end{align}
where the supremum is taken over all $G=\{G_{y}\}_{y\in Y_m}$ in $\mathrm{Meas}(Y_m,\cM_{C})$. It holds for every $y\in Y_m$ and $0\leq G_{y}\leq \idty$ that 
\begin{align}
|\omega_{AC}(F_{y}G_{y})-\omega_{AC}(M_{n}F_{y}M_{n}G_{y})|&\leq|\omega_{AC}(F_{y}G_{y}(\idty-M_n))|+|\omega_{AC}((\idty-M_n)F_{y}G_{y}M_n)|\\
&\leq2\sqrt{\omega_{AC}((\idty-M_n)^2)}\ ,  
\end{align}
where we used the Cauchy-Schwarz inequality for states, that is, $\omega(ab)^2\leq\omega(a^{*}a)\omega(b^{*}b)$ for any operators $a,b$. Hence, we have that the functionals $\omega_{C}^{n,y}(\cdot)=\omega_{AC}(M_{n}F_{y}M_{n}\cdot)$ converge uniformly to $\omega^{y}_{C}(\cdot)=\omega_{AC}(F_{y}\cdot)$ on the unit ball of $\cM_C$ for any $y \in Y_m$ (since $M_n$ converges in the $\sigma$-weak topology to $\idty$). Because the set $Y_m$ is finite, this also implies that $\omega_{Y_{m}C}^n=(\omega_{C}^{n,y})_{y\in Y_{m}}$ converges uniformly to $\omega_{Y_{m}C}=(\omega_{C}^{y})_{y\in Y_{m}}$ on the unit ball of $\cM_{Y_{m}C}$. Hence, we can interchange the limit for $n \rightarrow \infty$ with the supremum over $\mathrm{Meas}(Y_m,\cM_{A})$ and obtain  
\begin{align}
H_{\max}(X|B)_{\omega}\geq-H_{\min}(Y_m|C)_{\omega}-\log\sup_{x\in X,y\in Y}\left\|E_{x}^{1/2} F_{y}^{1/2}\right\|^{2}\ .
\end{align}
In a final step we take the infimum over all $m\in \mathbb N$ which gives the desired inequality due to the definition of the conditional min-entropy in~\eqref{eq:minapprox}. 
\end{proof}

Using the discrete approximation of the differential conditional min- and max-entropy from Proposition~\ref{thm:MinMaxApprox}, we obtain the uncertainty relation for continuous outcome measurements.

\begin{theorem}\label{thm:minmaxcont}
Let $\cM_{ABC}$ be a tripartite von Neumann algebra, $\omega_{ABC}\in\cS(\cM_{ABC})$ and $E_{X}\in\mathrm{Meas}(X,\cM_{A})$ and $F_{Y}\in\mathrm{Meas}(Y,\cM_{A})$ with coarse grained measure spaces $(X,\Sigma_X,\mu_X,\{\cP_\alpha\})$ and $(Y,\Sigma_Y,\mu_Y,\{\cQ_\beta\})$. If for the post-measurement states $\omega_{XBC}=\omega_{ABC}\circ E_{X}$ and $\omega_{YBC}=\omega_{ABC}\circ F_{Y}$, there exists an $\alpha_{0}>0$ such that $H_{\max}(X_{\cP_{\alpha_{0}}})_\omega<\infty$, then we have that
\begin{align}
h_{\max}(X|B)_{\omega}+h_{\min}(Y|C)_{\omega}\geq -\log c(E_{X},F_{Y})\ ,
\end{align}
where the overlap of the measurements is given by
\begin{align}\label{eq:ComplConstCont}
c(E_{X},F_{Y})=\liminf_{\alpha,\beta \rightarrow 0} \sup_{I\in\cP_\alpha, J\in \cQ_\beta}\frac{\|(E_X(\chi_I))^{1/2}\cdot(F_Y(\chi_J))^{1/2}\|^{2}}{\alpha\cdot\beta}\ ,
\end{align}
where $E_X(\chi_I)$ and $F_Y(\chi_J)$ are defined as in~\eqref{eq:DiscMeasurement}.
\end{theorem}

A similar relation derived under stronger conditions and with different techniques can be found in the thesis of one of the authors~\cite{FurrerPhD}. Note that in the case $X=\mathbb R$ the assumption that there exists an $\alpha_{0}>0$ such that $H_{\max}(X_{\cP_{\alpha_{0}}})_\omega<\infty$ is satisfied if the second moment of the distribution of $X$ is finite, or equivalently, the expectation value of the observable $\int x^2 E_{X}(x)dx $ with respect to $\omega_A$ is finite (see Lemma~\ref{lem:MaxCond}).

\begin{proof}
We first apply the uncertainty relation for measurements with a discrete number of outcomes (Proposition~\ref{prop:minmaxdisc}) to obtain for any partitions $\cP_\alpha$ and $\cQ_\beta$ the inequality
\begin{align}
\Big(H_{\max}(X_{\cP_{\alpha}}|B)_{\omega}+\log\alpha\Big)+\Big(H_{\min}(Y_{\cP_{\beta}}|C)_{\omega}+\log\beta\Big)\geq-\log \sup_{k,l} \frac{\|(E_{k}^{\cP_{\alpha}})^{1/2}\cdot(F_{l}^{\cP_{\beta}})^{1/2}\|^{2}}{\alpha \beta} \ ,
\end{align}
where the supremum in the logarithm is taken over all possible measurement operators $E_{k}^{\cP_{\alpha}}=E_X(\chi_{I^\alpha_k})$ and $F_{l}^{\cQ_{\beta}} = F_Y(\chi_{J^\alpha_l})$ with $\cP_\alpha=\{I^\alpha_k\}$ and $\cP_\beta=\{J^\alpha_l\}$. Taking the limit superior for $\alpha,\beta \rightarrow 0$ on both sides, we obtain the desired uncertainty relation by means of Proposition~\ref{thm:MinMaxApprox}.
\end{proof}


\subsection{Uncertainty Relations in Terms of the Quantum Conditional von Neumann Entropy}\label{sec:vNrelations}

We follow the same strategy as in the case of conditional min- and max-entropy and start with countably many outcomes. The following uncertainty relation for conditional von Neumann entropy was first derived in the finite-dimensional setting in Ref.~\onlinecite{Berta09}. 

\begin{proposition}\label{prop:URvNdisc}
Let $\cM_{ABC}$ be a tripartite von Neumann algebra, $\omega_{ABC}\in\cS(\cM_{ABC})$, $X$ and $Y$ countable, and $E_{X}=\{E_x\}_{x\in X} \in\mathrm{Meas}(X,\cM_{A})$ and $F_{Y}=\{F_y\}_{y \in Y} \in\mathrm{Meas}(Y,\cM_{A})$. Then, we have that
\begin{align}\label{eq,prop:URdiscVN}
H(X|B)_{\omega}+H(Y|C)_{\omega}\geq -\log c(E_{X},F_{Y})\ ,
\end{align}
where the overlap is given in~\eqref{eq:ccountable}. 
\end{proposition}

\begin{proof}
The result is obtained by following the same steps as in the proof of the statement for the conditional min- and max-entropy (Proposition~\ref{prop:minmaxdisc}). Doing so, one has to replace the conditional min- and max-entropy by the conditional von Neumann entropy and the max-relative entropy by the quantum relative entropy, respectively. Again, we first assume that $E_X$ and $F_Y$ are sub-normalized measurements with a finite number of outcomes $X$ and $Y$.

In the following, we use the same notation as in the proof of Proposition~\ref{prop:minmaxdisc}. By a similar argument as in the case of the conditional min- and max-entropy, the self duality of the conditional von Neumann entropy (Lemma~\ref{lem:duality_vN}) leads to
\begin{equation}
H(X|B)_{\omega} = D(\psi_{ACD}\circ V\|\tau_{X}\otimes\tilde\omega^V_{X'ACD})
\end{equation} 
with $\tilde\omega^V_{X'ACD}$ the restriction of $\psi_{ACD}\circ V$ onto $\mathbb C^{X'}\otimes \cM_{ABC}$. As in the case of the min- and max-entropy, the goal of the next step is to undo the dilated non-normalized measurement by applying $V^*$. But before using the monotonicity of the relative entropy under sub-unital maps, we first apply the projector $\Pi=\sum_{x}\ket {x,x} \bra{x,x}$ on $\mathbb C^{X}\otimes \mathbb C^{X'}$ to ensure that the second argument in the relative entropy is sub-normalized. Denoting by $\text{T}_O$ the transformation $a\mapsto  O^* a O$ with a suitable operator $O$, we can use the monotonicity of the quantum relative entropy under application of channels (Lemma~\ref{lem:mono:subunital}) to get
\begin{align}
D(\psi_{ACD}\circ V\|\tau_{X}\otimes\tilde\omega^V_{X'ACD})  & \geq  D(\psi_{ACD}\circ V\circ(\text{T}_\Pi +\text{T}_{\idty-\Pi} )   \|\tau_{X}\otimes\tilde\omega^V_{X'ACD}\circ(\text{T}_\Pi +\text{T}_{\idty-\Pi} )) \, .
\end{align}
Note now that the channel $\text{T}_\Pi +\text{T}_{\idty-\Pi}$ projects $\mathbb C^{X}\otimes \mathbb C^{X'}$ onto two orthogonal subspaces such that the term on the right hand side can be written as 
\begin{align}
D(\psi_{ACD}\circ V\circ\text{T}_\Pi   \|\tau_{X}\otimes\tilde\omega^V_{X'ACD}\circ\text{T}_\Pi) +  D(\psi_{ACD}\circ V\circ\text{T}_{\idty-\Pi}   \|\tau_{X}\otimes\tilde\omega^V_{X'ACD}\circ\text{T}_{\idty-\Pi} ) \, .
\end{align}
Since $\psi_{ACD}\circ V\circ\text{T}_{\idty-\Pi} = 0$, we have that the right term in the above sum is zero, and thus,
\begin{align}
D(\psi_{ACD}\circ V\|\tau_{X}\otimes\tilde\omega^V_{X'ACD})  \geq D(\psi_{ACD}\circ V    \|\tau_{X}\otimes\tilde\omega^V_{X'ACD}\circ\text{T}_\Pi) \, .
\end{align}
In a next step, we apply the sub-unital map $\cE(a)=VaV^*$ and get according to Lemma~\ref{lem:mono:subunital}
\begin{align}
D(\psi_{ACD}\circ V    \|\tau_{X}\otimes\tilde\omega^V_{X'ACD}\circ\text{T}_\Pi)  & \, \geq  D(\omega_{ABC}^V   \|\tau_{X}\otimes\tilde\omega^V_{X'ACD}\circ\text{T}_\Pi\circ \cE ) 
 + \omega_A(M-M^2)\log \omega_A(M-M^2) 
\, , 
\end{align}
where we used that $\tau_{X}\otimes\tilde\omega^V_{X'ACD}\circ\text{T}_\Pi(\idty) \leq 1$. Using again the monotonicity of the quantum relative entropy under sub-unital maps (Lemma~\ref{lem:mono:subunital}), the restriction onto the systems $AC$ followed by the application of the measurement $F_Y$ leads to the bound  
\begin{align}
 D(\omega_{ABC}^V   \|\tau_{X}\otimes\tilde\omega^V_{X'ACD}\circ\text{T}_\Pi\circ \cE )  \geq D(\omega_{YC}^V   \|\tau_{X}\otimes\tilde\omega^V_{X'ACD}\circ\text{T}_\Pi\circ \cE \circ F_Y ) + \omega^V_A(\idty-N)\log\omega^V_A(\idty-N)
\, , 
\end{align}
with $N=\sum_y F_y$. A straightforward computation similar to~\eqref{eq:boundC} shows that
\begin{align}
\tau_{X}\otimes\tilde\omega^V_{X'ACD}\circ\text{T}_\Pi\circ \cE \circ F_Y \leq c(E_X,F_Y) \tau_Y\otimes \tilde\omega^V_C\ ,
\end{align}
from which by basic properties of the quantum relative entropy (Lemma~\ref{lem:petz1} and Lemma~\ref{lem:petz2}), we obtain the bound   
\begin{align}
 D(\omega_{YC}^V   \|\tau_{X}\otimes\tilde\omega^V_{X'ACD}\circ\text{T}_\Pi\circ \cE \circ F_Y ) \geq D(\omega_{YC}^V   \|\tau_{Y}\otimes\tilde\omega^V_C) - \omega_{YC}^V(\idty)\log c(E_X,F_Y) 
\, .
\end{align}
Plugging all the steps together we finally arrive at 
\begin{align}\label{eq:ineqSubNormal}
H(X|B)_{\omega}  & \geq  D(\omega_{YC}^V   \|\tau_{Y}\otimes\tilde\omega^V_C) - \omega_{YC}^V(\idty)\log c(E_X,F_Y) \\
& \hspace{0.4cm}+ \omega_A(M-M^2)\log \omega_A(M-M^2)  + \omega^V_A(\idty-N)\log\omega^V_A(\idty-N)\label{eq:ineqSubNormal2}\ . 
\end{align}
Note that the above inequality reduces to~\eqref{eq,prop:URdiscVN} if both measurements $E_X$ and $F_Y$ are normalized. This can easily be seen by using that in such a case $M=N=\idty$, and thus,  $\omega^V_{YC} = \omega_{YC}$ and $\tilde\omega^V_C = \omega_C$.  

We now use the inequality~\eqref{eq:ineqSubNormal}-\eqref{eq:ineqSubNormal2} in a similar way as in the proof of Theorem~\ref{thm:minmaxcont} to obtain~\eqref{eq,prop:URdiscVN} for an infinite number of outcomes. For that we let $X_n$ and $Y_n$, $n\in \mathbb N$, as well as $E_{X_n}$ and $F_{Y_n}$ be as in the proof of Proposition~\ref{thm:minmaxcont}. For fixed sub-normalized measurements $E_{X_n}$ and $F_{Y_m}$, the inequality~\eqref{eq:ineqSubNormal}-\eqref{eq:ineqSubNormal2} then reads 
\begin{align}
H(X_n|B)_{\omega}  & \geq  D(\omega_{Y_m C}^n   \|\tau_{Y_m}\otimes\tilde\omega^n_C) - \omega_{Y_mC}^n(\idty)\log c(E_{X},F_{Y}) \\
& \hspace{0.4cm}+ \omega_A(M_n-M_n^2)\log \omega_A(M_n-M_n^2)  + \omega^n_A(\idty-N_m)\log\omega^V_A(\idty-N_m) 
\, .
\end{align}
Here, we used the same notation as in the proof of Proposition~\ref{prop:minmaxdisc} and set further $\tilde\omega^n_C(a) = \omega(M_n a)$ as well as $N_m =\sum_{y\in Y_m} F_y$. Let us take the limit inferior for $n,m\rightarrow \infty$. According to the definition of the conditional von Neumann entropy (Definition~\ref{def:condneumann}) we have that $ H(X_n|B)_{\omega^n}  $ converges to $ H(X|B)_{\omega}$. Furthermore, we use the lower semi-continuity of the quantum relative entropy (see, e.g., Ref.~\onlinecite[Corollary 5.12]{PetzBook}) and that $\omega_{Y_m C}^n$ and $\tilde\omega^n_C$ converge to $\omega_{YC}$ and $\omega_C$, respectively, to get that $\liminf_{n,m} D(\omega_{Y_m C}^n   \|\tau_{Y_m}\otimes\tilde\omega^n_C)  \geq H(Y|C)_\omega$. Using that $M_n$ and $N_m$ converge $\sigma$-weakly to $\idty_A$, it is straightforward to see that $ \omega_{Y_mC}^n(\idty) \rightarrow 1$, $\omega_A(M_n-M_n^2)\rightarrow 0$ as well as $\omega^n_A(\idty-N_m)\rightarrow 0$. Using that $x\log x \rightarrow 0 $ for $x\rightarrow 0$, we finally obtain~\eqref{eq,prop:URdiscVN}.

\end{proof}

\begin{theorem}\label{prop:vN_tri_disc}
Let $\cM_{ABC}$ be a tripartite von Neumann algebra, $\omega_{ABC}\in\cS(\cM_{ABC})$ and $E_{X}\in\mathrm{Meas}(X,\cM_{A})$ and $F_{Y}\in\mathrm{Meas}(Y,\cM_{A})$ with coarse grained measure spaces $(X,\Sigma,\mu_X,\{\cP_\alpha\})$ and $(Y,\Sigma_Y,\mu_Y,\{\cQ_\beta\})$. If the post-measurement states $\omega_{XBC}=\omega_{ABC}\circ E_{X}$ and $\omega_{YBC}=\omega_{ABC}\circ F_{Y}$ satisfy $-\infty<h(X|B)_{\omega}$, $-\infty<h(Y|C)_{\omega}$, and if there exists $\alpha_0> 0$ for which $H(X_{\cP_{\alpha_0}} |B)_{\omega}<\infty$ as well as $\beta_0>0$ for which $H(Y_{\cQ_{\beta_0}}|C)_{\omega}<\infty$, then it holds that
\begin{align}
h(X|B)_{\omega}+h(Y|C)_{\omega}\geq-\log c(E_{X},F_{Y})\ ,
\end{align}
where $c(E_{X},F_{Y})$ is as in~\eqref{eq:ComplConstCont}.
\end{theorem}

The theorem is obtained via the approximation result for the differential conditional von Neumann entropy (Proposition~\ref{prop:vNapprox}) using the exactly same steps as in the proof of the corresponding result for the differential conditional min- and max-entropy (Theorem~\ref{thm:minmaxcont}).


\subsection{Entropic Uncertainty Relations for Position and Momentum Operators}\label{sec:posmomrelations}

Let $Q$ and $P$ be a pair of position and momentum operators defined via the canonical commutation relation $[Q,P]= i$ where we set $\hbar = 1$. The unique representation space is $\cH=L^2(\mathbb R)$, with $Q$ the multiplication operator and $P$ the first order differential operator. Both operators possess a spectral decomposition with a positive operator valued measure $E_Q$ and $E_P$ in $\rm{Meas}(\mathbb R,\cB(\cH))$.

Let us start with finite spacing measurements and assume that the precision of the position and momentum measurement are given by intervals of length $\delta_q$ and $\delta_p$ for the entire range of the spectrum. As we will see later, the overlap in the uncertainty relation only depends on the spacings $\delta_q$ and $\delta_p$ but not on the explicit coarse grainings $\cQ_{\delta_q}=\{I_{\delta_q}^k\}_{k=1}^\infty$ and $\cP_{\delta_p}=\{J_{\delta_p}^k\}_{k=1}^\infty$. We can thus fix arbitrary coarse grainings $\cQ_{\delta_q}$ and $\cP_{\delta_p}$ for given spacings $\delta_q$ and $\delta_p$. The corresponding measurements are then formed by the positive operators $Q^k[\delta_q] = E_Q(I_{\delta_q}^k)$ and $P^k[\delta_p] = E_P(J_{\delta_p}^k)$. Following the notation in Section~\ref{sec:Results}, we denote the discrete classical systems induced by these position and momentum measurements by $Q_{\delta_q}$ and $P_{\delta_p}$. In order to keep the notation simple, we omit the dependence of the distributions on the particular coarse grainings. 

According to Proposition~\ref{prop:minmaxdisc} and~\ref{prop:URvNdisc}, the quantity which enters the entropic uncertainty relation is the overlap of the measurement operators
\begin{align}\label{eq:DefQPconstDisc}
\sup_{k,l} \Vert \sqrt{ Q^k[\delta_q]} \sqrt{ P^l[\delta_p]} \Vert^2 = \sup_{k,l} \Vert Q^k[\delta_q]  P^l[\delta_p] Q^k[\delta_q] \Vert\ .
\end{align}
Note first that $\Vert Q^k[\delta_q]  P^l[\delta_p] Q^k[\delta_q] \Vert $ can only depend on the length of the intervals and is similar for any $k$ and $l$, which are taken to be $k=l=1$ in the following. The reason is that the translation in position and momentum space are given by the unitary transformations $\exp[-i x_0 P]$ and $\exp[-i p_0 Q]$, which leave the norm invariant. Furthermore, since dilation operators are unitary, the constant only depends on the invariant product $\delta q \delta p$. The operator $H(\delta q,\delta p)=Q^1[\delta_q]P^1[\delta_p]Q^1[\delta_q]$ is important for time-limiting and lowpassing signals, and its largest eigenvalue, and thus its norm, can be expressed by~\cite{Slepian1964} (see also, e.g., Ref.~\onlinecite{Kiukas10} and references therein)
\begin{align}\label{eq:QPconstDisc}
c(\delta q,\delta p) =\frac{1}{2\pi}\delta q\delta p\cdot S_{0}^{(1)}\left(1,\frac{\delta q\delta p}{4}\right)^{2}\ ,
\end{align}
where $S_{0}^{(1)}(1,\cdot)$ denotes the 0th radial prolate spheroidal wave function of the first kind. For $\delta q\delta p \rightarrow 0$, it follows that $S_{0}^{(1)}\left(1,\frac{\delta q\delta p}{4}\right)\rightarrow 1$, such that the overlap behaves as $ c(\delta q,\delta p) \approx \frac{\delta q\delta p}{2\pi}$ for small spacing. A plot of the overlap $c(\delta q,\delta p)$ as well as the complementarity constant $-\log c(\delta q,\delta p)$ are shown for $\delta q = \delta p$ in Fig.~\ref{fig}.

\begin{figure}
\includegraphics{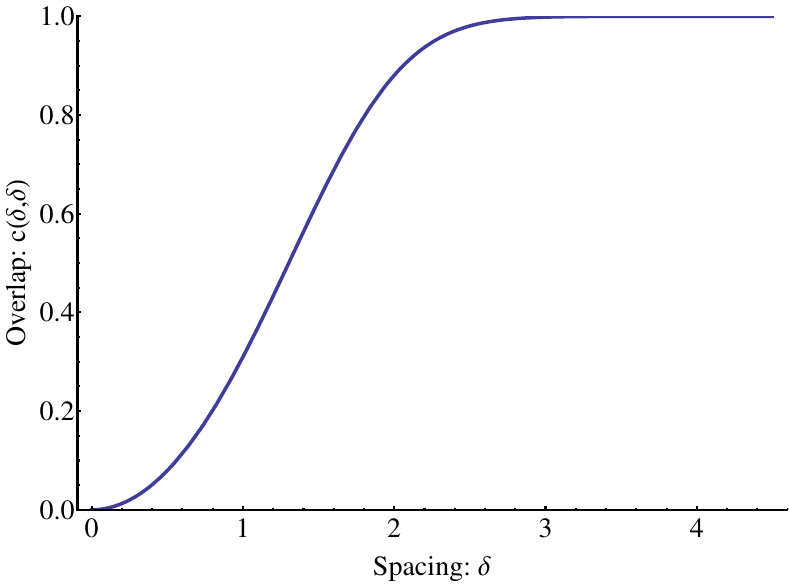}
\includegraphics{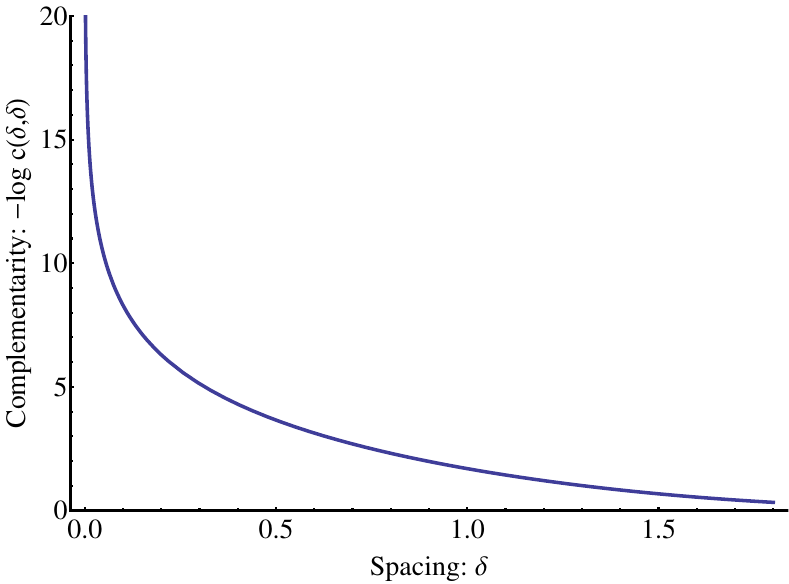}
\caption{The plot	on the left hand side shows the overlap $c(\delta q,\delta p)$ depending on $\delta$ where $\delta^2 = \delta q \delta p$. One can see that for $\delta \geq 3$ the value is already approximately $1$ and the uncertainty relation gets trivial. The plot on the right hand side shows the behavior of the complementarity constant $-\log c(\delta q,\delta p)$ for small $\delta = \sqrt{\delta q \delta p}$. Note that we have set $\hbar =1$ such that the minimal uncertainty product of the standard deviations is $\sqrt{\text{Var}(P)} \sqrt{\text{Var}(Q)} = 1/2$ and the vacuum has a variance of $1/2$.}
\label{fig} 
\end{figure}

\begin{corollary}\label{cor:pqfinite}
Let $\cM_{ABC}=\cB(L^{2}(\mathbb{R}))\otimes\cM_{BC}$ with $\cM_{BC}$ a von Neumann algebra, and consider position and momentum measurements with spacing $\delta_q> 0$ and $\delta_p> 0$ on system $A$. Then, we have that
\begin{align}
H_{\max}^{}(Q_{\delta_q}|B)_{\omega}+H_{\min}^{}(P_{\delta_p} |C)_{\omega}\geq - \log c(\delta_q,\delta_p)\label{eq:minmax_pq_disc}\ ,
\end{align}
and
\begin{align}
H(Q_{\delta_q}|B)_{\omega}+H(P_{\delta_p}|C)_{\omega}\geq -\log c(\delta_q,\delta_p)\label{eq:neumann_pq_disc}\ ,
\end{align}
where $c(\delta_q,\delta_p)$ is given in~\eqref{eq:QPconstDisc}.
\end{corollary}

The corollary follows directly from Proposition~\ref{prop:minmaxdisc}. Since the statement is invariant under exchanging $Q$ and $P$, the uncertainty relation in~\eqref{eq:minmax_pq_disc} also holds for the conditional min-entropy of $\omega_{Q_{\delta_q}B}$ and the conditional max-entropy of $\omega_{P_{\delta_p}C}$. Corollary~\ref{cor:pqfinite} complements known results for the Shannon entropy~\cite{Partovi83,Birula84,Rudnicki10,Rudnicki11,PhysRevA.85.042115} and for the R\'enyi entropy (for the order pair $\infty-1/2$, cf.~\eqref{eq:renyiinfty} and~\eqref{eq:renyi12})~\cite{Birula06,Rudnicki10,PhysRevA.85.042115} in the presence of quantum memories.

Let us address the sharpness of the uncertainty relations. More precisely, we say that an uncertainty relation is sharp if there exists a state for which equality is attained.

\begin{proposition}\label{lem:tight}
The entropic uncertainty relation in terms of the conditional min- and max-entropy in~\eqref{eq:minmax_pq_disc} is sharp for any spacing $\delta_q$ and $\delta_p$.
\end{proposition}

\begin{proof}
By the data-processing inequalities for the conditional min- and max-entropy (Proposition~\ref{prop:data}), it is enough to show sharpness for
\begin{align}\label{eq:tight}
H_{\max}(Q_{\delta_q})_{\omega}+H_{\min}(P_{\delta_p})_{\omega}\geq\log c(\delta_q,\delta_p)\ .
\end{align}
Let us assume that the partitions are centralized around $0$, such that they contain the intervals $I=[-\delta_q/2,\delta_q/2]$ and $J=[-\delta_p/2,\delta_p/2]$, respectively. We take a normalized pure state given by a wave function $\psi(q)\in L^{2}(\mathbb{R})$ with support on $I$. It then follows that the distribution of the discretized position measurement is peaked and thus, $H_{\max}(Q(\delta_q))_{\omega}=0$. The probability distribution of the momentum measurement is given by $|\cF[\psi](p)|^{2}$, where $\cF$ denotes the Fourier transform. Therefore, we find for the min-entropy that
\begin{align}
2^{H_{\min}(P_{\delta_p})_{\omega}}\leq \int\chi_{J}(p)\big|\cF[\psi](p)\big|^{2}{\rm d}p & =\frac{1}{2\pi}\int\chi_{I}(q_1)\chi_{J}(p)\chi_{I}(q_2)\bar{\psi}(r)\psi(q)e^{-i\left(q-r\right)p}{\rm d}q_1{\rm d}p{\rm d}q_2 \\ & = \braket{\psi}{Q[I]P[J]Q[I] \psi}\ ,
\end{align}
where $\chi_I$ denotes the indicator function on I. But also note that the overlap is given by
\begin{align}
\left\|Q[I]P[J]\right\|^2=\left\|Q[I]P[J]Q[I]\right\| =\sup_{\phi\in L^{2}(X)}\braket{\phi}{Q[I]P[J]Q[I] \phi}\ .
\end{align}
Since the supremum over $\phi$ can be restricted to functions with support on $I$, we find the claim by choosing $\psi$ as the optimal $\phi$. We finally note that this state $\psi(q)$ is given by the normalized projection of the radial prolate spheroidal wave function of the first kind onto the interval $I$ (see, e.g., Ref.~\onlinecite{Kiukas10} and references therein).
\end{proof}

The sharpness for the uncertainty relation in terms of the conditional von Neumann entropy ~\eqref{eq:neumann_pq_disc} is a more difficult question, since one might expect that a non-trivial quantum memory is crucial. This is because without quantum memory the uncertainty relation 
\begin{align}
H(Q_{\delta_q})_{\omega}+H(P_{\delta_p})_{\omega}\geq\log\left(\frac{e\pi}{\delta_q\delta_p}\right)\ ,
\end{align}
with a different lower bound has been shown~\cite{Birula84}, which becomes a better bound for small enough $\delta_q\delta_p$.

Let us now consider continuous position-momentum distributions. In order to compute the measurement overlap given in~\eqref{eq:ComplConstCont}, we can simply take the limit of $c(\delta,\delta)$ for $\delta\rightarrow 0$ yielding 
\begin{align}
c(E_Q,E_P) =\lim_{\delta\rightarrow 0}\frac1{2\pi} \ S_{0}^{(1)}\left(1,\frac{\delta^2}{4}\right)^{2}=\frac1{2\pi}\ ,
\end{align}
where we used~\eqref{eq:QPconstDisc}, and that $S_{0}^{(1)}\left(1,\frac{\delta^2}{4}\right)\rightarrow 1$ for $\delta\rightarrow 0$. 
Hence, we immediately obtain the following corollary.

\begin{corollary}\label{cor:cont_pq}
Let $\cM_{ABC}=\cB(L^{2}(\mathbb{R}))\otimes\cM_{BC}$ with $\cM_{BC}$ a von Neumann algebra, $\omega_{ABC}\in\cS(\cM_{ABC})$, and denote the post-measurement states obtained by continuous position and momentum measurements on system $A$ by $\omega_{QBC}$ and $\omega_{PBC}$. If there exists a finite spacing $\delta_q$ such that $H_{\max}(Q_{\delta_q})_{\omega}<\infty$, then we have that
\begin{align}\label{eq:minmax_pq_cont}
h_{\min}(Q|B)_{\omega}+h_{\max}(P|C)_{\omega}\geq\log2\pi\ .
\end{align}
Furthermore, if $-\infty<h(Q|B)_{\omega}$, $-\infty<h(P|C)_{\omega}$, and if there exists finite spacings $\delta_q , \delta_p$ for which $H(Q_{\delta_q}|B)_\omega<\infty$ and $H(P_{\delta_p}|C)_\omega<\infty$, then we have that
\begin{align}\label{eq:PQcontURvN}
h(Q|B)_{\omega}+h(P|C)_{\omega}\geq\log2\pi\ .
\end{align}
\end{corollary}

Let us first note that for states with finite expectation for the operator $Q^2 + P^2$ we can always find a spacing for which $H_{\max}(Q_{\delta_q})_{\omega}$, $H(Q_{\delta_q}|B)_\omega$ and $H(P_{\delta p}|C)_\omega$ are less than $\infty$. This follows from Lemma~\ref{lem:MaxCond}, which says that the condition is satisfied whenever $\omega_A(P^2)$ and $\omega_A(Q^2)$ are finite. Note that due to the data processing inequality and the fact that max-entropy is larger as the von Neumann entropy, the latter is as well bounded whenever the assumptions of Lemma~\ref{lem:MaxCond} are satisfied. Hence, if considering modes of an electromagnetic field, the uncertainty relation for the conditional min-and max-entropy holds for any state with finite mean energy, while for the conditional von Neumann entropy the only further assumption is that $h(Q|B)_{\omega}$ and $h(P|C)_{\omega}$ are not $-\infty$.   

The uncertainty relation for the differential min- and max-entropy~\eqref{eq:minmax_pq_cont} is already sharp without quantum memory~\cite{Birula06}. The minimal uncertainty states are pure Gaussian states, where the product of the variances of the position and momentum measurements are minimal, that is, $\text{Var}(Q)\text{Var}(P)= \frac14$. This follows from the fact that for a Gaussian distribution $X$ with variance $\sigma$, $H_{\min}(X)=\log \sqrt{2 \pi \sigma}$ and $H_{\max}(X)=\log 2\sqrt{2 \pi \sigma}$. However, the uncertainty relation for the von Neumann entropy~\eqref{eq:PQcontURvN} is not sharp in the case of no quantum memory. This is due to the same reason as already encountered in the discrete case. Namely, the tighter inequality 
\begin{align}\label{eq:hirschmann}
H(Q)_{\omega}+H(P)_{\omega}\geq\log e\pi
\end{align}
holds in the absence of quantum memories~\cite{Beckner75,Birula75}. Another sharp uncertainty relation without quantum memory has recently been shown by Lieb and Frank~\cite{Lieb11}: $H(Q)_{\omega}+H(P)_{\omega}\geq\log 2\pi+H(A)_{\omega}$ (see also Ref.~\onlinecite{Rumin11,Rumin12}). But strikingly, the uncertainty relation~\eqref{eq:PQcontURvN} is sharp if we include quantum memory. In particular, take $\cM_B=\cB(L^{2}(\mathbb{R}))$ and $\cM_C$ trivial, then the EPR state~\cite{Weedbrook12} on $AB$ for infinite squeezing saturates inequality~\eqref{eq:PQcontURvN}. Note that the EPR state is a Gaussian state with covariance matrix
\begin{align} 
\Gamma^{AB}(\nu) = \frac{1}{2} 
\left(
\begin{array}{cc}
\nu \idty_2 & \sqrt{\nu^2 -1} Z \\
\sqrt{\nu^2 -1} Z & \nu \idty_2 \\
\end{array}
\right)\ ,
\end{align}
where $\nu=\cosh(2 r)$ with $r$ the squeezing strength and $Z=\text{diag}(1,-1)$. The covariance matrix is written with respect to a phase space parametrization given by $(q_A,p_A,q_B,p_B)$.  In the following we denote by $\omega_{AB}^\nu$ the Gaussian state corresponding to $\Gamma^{AB}(\nu)$. The variance of the outcome distribution of the $P$ measurement on the $A$ system (which is Gaussian as well) is given by $\Gamma_{2,2}^{AB}(\nu) = \nu/2$, and thus,
\begin{align}
h(P)_{\omega^\nu}  = \log(e)/2 + \log\sqrt{\pi \nu}\ .
\end{align}
In order to compute $h(Q|B)_{\omega^\nu}$, we first note that since $ h(Q)_{\omega^\nu}  < \infty$ we can use Proposition~\ref{prop:vNapprox} to write
\begin{align}
h(Q|B)_{\omega^\nu} = h(Q)_{\omega^\nu} - D(\omega^\nu_{QB}|| \omega_Q \otimes \omega_B)\ .
\end{align}
Due to $\Gamma_{1,1}^{AB}(\nu) = \Gamma_{2,2}^{AB}(\nu) $, we get $h(Q)_{\omega^\nu} = h(P)_{\omega^\nu} $. By using disintegration theory, Ref.~\onlinecite[Chapter IV.7]{Takesaki1}, we can further compute 
\begin{align}
D(\omega^\nu_{QB}|| \omega_Q \otimes \omega_B) = \int \omega_Q(x) \tr( \omega_B^x \log \omega_B^x) dx +  h(Q)_{\omega^\nu}   - h(Q)_{\omega^\nu}   + H(B)_{\omega^\nu}\ ,
\end{align} 
where $\omega_B^x$ is the normalized post-measurement state on $B$ conditioned on the outcome $x$ of the $Q$ measurement on $A$. Note that for every $x$ the state $\omega_B^x$ is a pure Gaussian states with a covariance matrix independent on $x$ and a displacement $x$ (see, e.g., Ref.~\onlinecite{Lodewyck07}). Thus, we end up with $D(\omega^\nu_{QB}|| \omega_Q \otimes \omega_B) =  H(B)_{\omega^\nu}$. Note that in the case of a Gaussian state the von Neumann entropy only depends on the symplectic eigenvalues of the covariance matrix~\cite{Holevo99,Weedbrook12}, and in our case it is given by $ H(B)_{\omega^\nu} =  t \log t - (t-1)\log(t-1)$ with $t=(\nu +1)/2$. Hence, we finally get a closed expression for the left hand side of~\eqref{eq:PQcontURvN}
\begin{align}
f(\nu) = \log(e\pi \nu) - \frac{\nu+1}{2} \log  \frac{\nu+1}{2} + \frac{\nu-1}{2} \log  \frac{\nu-1}{2}\ .
\end{align} 
Note that $f(\nu) \rightarrow \log(2\pi)$ for $\nu \rightarrow \infty$ and the uncertainty relation~\eqref{eq:PQcontURvN} gets sharp. As illustrated in Figure~\ref{fig2}, the gap closes exponentially in the squeezing parameter $r$, and thus, linear in the energy. We conclude the above discussion with the following statement. 

\begin{figure}
\includegraphics{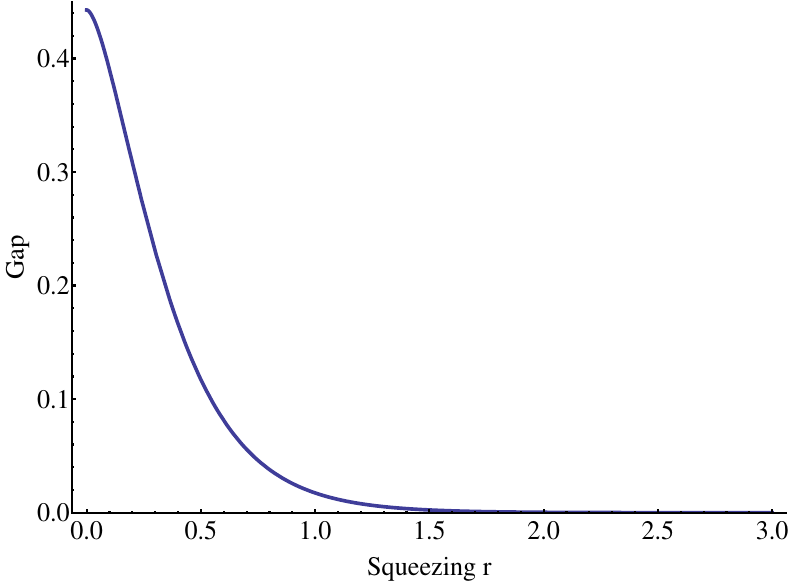}
\includegraphics{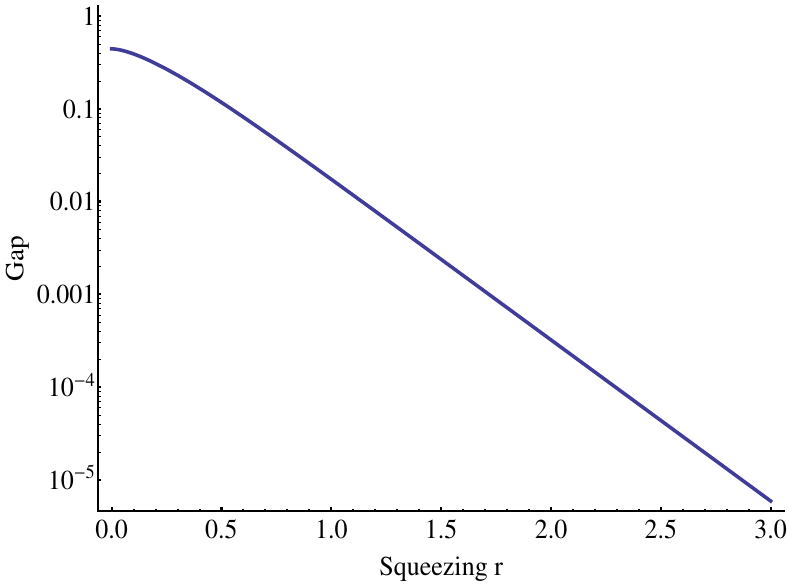}
\caption{The two plots show the gap $f(\nu) - \log(2\pi)$ in dependence on the squeezing $ r$ ($\nu =\cosh 2r$) for a linear (left) and logarithmic scaling (right). The plot on the right hand side illustrates the exponential convergence in $r$. Note that the mean energy of an EPR state with squeezing $r$ with respect to the harmonic oscillator Hamiltonian $Q_A^2 + P_A^2 +Q_B^2 + P_B^2$ is given by $1+ 2\sinh(r/2)^2$. Considering two-mode squeezed vacuum states of light, an experimentally achievable squeezing of $10$dB~\cite{Eberle13} corresponds to a squeezing of $r \approx 1.5$ for which the gap is already negligible.}
\label{fig2}
\end{figure}

\begin{proposition}\label{lem:tightvonNeum}
The entropic uncertainty relations for continuous position and momentum measurements stated in~\eqref{eq:minmax_pq_cont} and~\eqref{eq:PQcontURvN} are sharp.
\end{proposition}


\section{Conclusion and Outlook}\label{sec:outlook}

We have shown entropic uncertainty relations in the presence of quantum memory for states on von Neumann algebras and measurements with an infinite or continuous outcome range. Our relations are expressed in terms of differential quantum conditional von Neumann entropy and differential quantum conditional min- and max-entropy. We further established a discrete approximation theorem for those quantum conditional entropies by their discrete and regularized counterparts. We have shown that the uncertainty relations in the continuous case are sharp in the sense that there exists a state for which equality holds. Hence, they provide the best possible state-independent bounds. Moreover, it turned out that in the case of the von Neumann entropy, the minimal uncertainty bound is lowered in the presence of quantum memories.

The use of quantum conditional entropy measures to express the uncertainty principle is motivated by their importance in quantum information theory. Whereas von Neumann entropy based measures are the most studied in quantum physics and asymptotic quantum information theory, the conditional min- and max-entropy have applications in non-asymptotic quantum information theory and quantum cryptography. For example, an entropic uncertainty relation for finite and discrete conditional min- and max-entropy has been used to show security for discrete variable quantum key distribution~\cite{Tomamichel11,Lim11}. Similarly, our uncertainty relation for position and momentum observables and conditional min- and max-entropy in~\eqref{eq:minmax_pq_disc} and~\eqref{eq:minmax_pq_cont} have direct application to continuous variable quantum key distribution, and where the key ingredient for the first quantitative finite-key security analysis against coherent attacks~\cite{PhysRevLett.109.100502,furrer14}. Moreover, the uncertainty relation for the differential quantum conditional von Neumann entropy in~\cite{} has been used to obtain simple key rate formulas in the asymptotic limit~\cite{Walk14}. Finally, it was already suggested in~\cite{Berta09} to use bipartite entropic uncertainty relations in the presence of quantum memory for entanglement witnessing. This would be certainly interesting for continuous variable systems since it provides a simple criterion. Ideas in this direction have for instance been developed in~\cite{Wallborn11,Schneeloch12,Ray13}.

An interesting open question concerns the derivation of bipartite continuous variable entropic uncertainty relations. In the case of the conditional von Neumann entropy and rank-$1$ measurements, the tripartite uncertainty relation is equivalent to~\cite{Berta09}
\begin{align}\label{eq:URsideInfo_bi}
H(X|B)+H(Y|B)\geq\log\frac{1}{c}+H(A|B)\ , 
\end{align}
with the same constant $c$ as in the tripartite version. Recently, Frank and Lieb~\cite{Lieb12} extended the relation~\eqref{eq:URsideInfo_bi}, and proved that for any finite-dimensional bipartite quantum state $\rho_{AB}$ as well as finite measurements $\{E_{x}\}_{x\in X}$ and $\{F_{y}\}_{y\in Y}$ on $A$,
\begin{align}\label{eq:lieb}
H(X|B)+H(Y|B)\geq\log\frac{1}{c_{1}}+H(A|B)\ ,
\end{align}
where $c_1=\max_{x,y}\tr\big[E_x F_y\big]$. The constant $c_{1}$ agrees with the constant $c$ if at least one of the measurements is rank-one projective~\cite{Coles10}, and is otherwise an upper bound. Frank and Lieb then go on and extend this to continuous position and momentum measurements, for which they get $c_{1}=\frac{1}{2\pi}$. An alternative generalization of~\eqref{eq:URsideInfo_bi} (that is tighter than the relation~\eqref{eq:lieb} for some natural applications) was presented by one of the authors in his thesis~\cite{tomamichel:thesis}. This approach is motivated by the following gedankenexperiment. Consider a quantum system $A$ comprised of two qubits, $A_1$ and $A_2$, where $A_1$ is maximally entangled with a second system, $B$, and $A_2$ is in a fully mixed state, product with $A_1$ and $B$. We employ projective measurements $E^1$ and $F^1$ which measure $A_1$ in two mutually unbiased bases and leave $A_2$ intact. Analogously, $E^2$ and $F^2$ measure $A_2$ in mutually unbiased bases and leave $A_1$ intact. Evaluating the terms of interest for this setup yields $c(E^1,F^1)=c(E^2, F^2) = \frac{1}{2}$ and $c_1(E^1, F^1) = c_1(E^2, F^2) = 1$ as well as $H(A|B) = H(A_1|B) + H(A_2) = -1 + 1 = 0$. Indeed, if the maximally entangled system $A_1$ is measured, we find that $H(X|B) + H(Y|B) = 0$, and the bound by Frank and Lieb~\eqref{eq:lieb} is sharp. On the other hand, if $A_2$ is measured instead, we expect that $H(X|B) + H(Y|B) = 2$ and the bound is far from sharp. Examining the above example, it is clear that the expected uncertainty depends strongly on which of the two systems is measured and thus, how much entanglement is consumed. However, this information is not taken into account by the overlaps $c$ or $c_1$, nor by the entanglement of the initial state, $H(A|B)$. In the above example, it is straightforward to see that if $A_1$ ($A_2$) is measured, the average entanglement left in the post-measurement state measured by the von Neumann entropy is given by $H(A_2|B)$ ($H(A_1|B)$). Hence,~\eqref{eq:URsideInfo_bi} would turn into
\begin{align}
H(X|B)+H(Y|B)\geq\log\frac{1}{c}+\Big(H(A|B)-H(A'|B)\Big)\ ,
\end{align}
where $A'$ corresponds to $A_2$ if $A_1$ is measured and $A_1$ if $A_2$ is measured instead. It is easy to verify that the above inequality is sharp for both examples considered so far. This suggests that~\eqref{eq:URsideInfo_bi} can be generalized by considering the difference in entanglement of the state before and after measurement. The entanglement of the post-measurement state vanishes for rank-one projective measurements, which is why this contribution was initially overlooked\,---\,however, it must be accounted for when considering general measurements. An alternative bipartite entropic uncertainty relation can now be stated in the following way. Let $\{E_{x}\}_{x\in X}$ and $\{F_{y}\}_{y\in Y}$ be finite measurements on a finite-dimensional quantum system $A$, and denote the Stinespring dilations of $\{E_{x}\}_{x\in X}$ and $\{F_{y}\}_{y\in Y}$ by $U$ and $V$, respectively. Then, it holds for any finite-dimensional bipartite quantum state $\rho_{AB}$ that
\begin{align}\label{eq:tomamichel}
H(X|B)_{\omega\circ E}+H(Y|B)_{\omega\circ F}\geq-\log c(E_{X},F_{Y})+H(A|B)_{\omega}-\min\Big\{H(A|XB)_{\omega\circ V},H(A|YB)_{\omega\circ U}\Big\}\ .
\end{align}
It is possible to reformulate the relation~\eqref{eq:tomamichel} in terms of von Neumann mutual information or von Neumann conditional mutual information, and it would be interesting to find formulations of~\eqref{eq:tomamichel} that also hold for continuous measurements and infinite-dimensional quantum memory.

\section*{Acknowledgements}
We gratefully acknowledge discussions with Renato Renner, Reinhard F.~Werner and Jukka Kiukkas and express special thanks to Michael Walter. MB, VBS and MC acknowledge financial support by the Swiss National Science Foundation (grants PP00P2-128455, 20CH21-138799 (CHIST-ERA project CQC)), the Swiss National Center of Competence in Research 'Quantum Science and Technology (QSIT)', the Swiss State Secretariat for Education and Research supporting COST action MP1006 and the European Research Council under the European Union's Seventh Framework Programme (FP/2007-2013) / ERC Grant Agreement no. 337603. MC also acknowledges a Sapere Aude Grant from the Danish Council for Independent Research. VBS is in addition supported by an ETH Postdoctoral Fellowship. FF acknowledges support from Japan Society for the Promotion of Science (JSPS) by KAKENHI grant No. 24-02793. MT is funded by the Ministry of Education (MOE) and National Research Foundation Singapore, as well as MOE Tier 3 Grant "Random numbers from quantum processes" (MOE2012-T3-1-009).

\appendix

\section{Properties of the Conditional Min- and Max-Entropy}

\begin{definition}\label{def:maxrelative}
Let $\cM$ be a von Neumann algebra, $\omega\in\cP(\cM)$, and $\sigma\in\cP(\cM)$. Then, the max-relative entropy of $\omega$ with respect to $\sigma$ is defined as
\begin{align}
D_{\max}(\omega\|\sigma)=\inf\{\iota\in\mathbb{R}:\omega\leq2^{\iota}\cdot\sigma\}\ .
\end{align}
\end{definition}

\begin{definition}\label{lem:HminDmax}
Let $\cM_{AB}=\cB(\cH_{A})\otimes\cM_{B}$ with $\cH_{A}$ a finite-dimensional Hilbert space, $\cM_{B}$ a von Neumann algebra, and $\omega_{AB}\in\cS_{\leq}(\cM_{AB})$. Then, the conditional min-entropy of $A$ given $B$ is defined as
\begin{align}
H_{\min}(A|B)_{\omega}=-\inf_{\sigma_{B}\in\cS(\cM_{B})}D_{\max}(\omega_{AB}\|\tau_{A}\otimes\sigma_{B})\ ,
\end{align}
where $\tau_{A}$ denotes the trace on $\cB(\cH_{A})$. Furthermore, the conditional max-entropy of $A$ given $B$ is defined as
\begin{align}
H_{\max}(A|B)_{\omega}=\sup_{\sigma_{B}\in\cS(\cM_{B})}F(\omega_{AB},\tau_{A}\otimes\sigma_{B})\ .
\end{align}
\end{definition}

\begin{lemma}\label{lem:mon_Dmax}
Let $\cM_{A},\cM_{B}$ be von Neumann algebras, $\omega_{A},\sigma_{A}\in\cP(\cM_{A})$, and let $\cE:\cM_{B}\rightarrow\cM_{A}$ be a normal, completely positive, and sub-unital map. Then, we have that
\begin{align}
D_{\max}(\omega_{A}\|\sigma_{A})\geq D_{\max}(\omega_{A}\circ\cE\|\sigma_{A}\circ\cE)\ .
\end{align}
\end{lemma}

\begin{lemma}\label{lem:minmax_elementary1}
Let $\cM$ be a von Neumann algebra, and $\omega,\sigma\in\cP(\cM)$ with $\sigma\geq\gamma$. Then, we have that
\begin{align}
D_{\max}(\omega\|\gamma)\geq D_{\max}(\omega\|\sigma)\ .
\end{align}
\end{lemma}

\begin{lemma}\label{lem:minmax_elementary2}
Let $\cM$ be a von Neumann algebra, $\omega,\sigma\in\cP(\cM)$, and $c>0$. Then, we have that
\begin{align}
D_{\max}(\omega\|c\cdot\sigma)=D_{\max}(\omega\|\sigma)+\log1/c\ .
\end{align}
\end{lemma}

\begin{proposition}\label{prop:data}
Let $\cM_{XBC}=L^{\infty}(X)\otimes\cM_{BC}$ with $(X,\Sigma,\mu)$ a $\sigma$-finite measure space, $\cM_{BC}$ a bipartite von Neumann algebra, and $\omega_{XBC}\in\cS_{\leq}(\cM_{XBC})$. Then, we have that
\begin{eqnarray}\label{eq:DataProcessingIneq}
h_{\min}(X|BC)_{\omega} &\leq& h_{\min}(X|B)_{\omega} \\
h_{\max}(X|BC)_{\omega} &\leq& h_{\max}(X|B)_{\omega}\ .
\end{eqnarray}
\end{proposition}

\begin{proof}
The inequality for the differential conditional min-entropy is obtained by using that any $E\in\rm{Meas}(X,\cM_B)$ can be embedded into $\rm{Meas}(X,\cM_{BC})$ such that $\omega^x_{B}(E_x)=\omega^x_{BC}(E_x)$. For the differential conditional max-entropy, one exploits the fact that the fidelity can only increase under restrictions to a subsystem, that is, $F(\omega_{BC},\sigma_{BC})\leq F(\omega_{B},\sigma_{B})$ (as shown in~\cite{Alberti-1983}).
\end{proof}

\begin{lemma} (Ref.~\onlinecite[Proposition 4.14]{Berta11})\label{lem:duality}
Let $\cM_{AB}=\cB(\cH_{A})\otimes\cM_{B}$ with $\cH_{A}$ a finite-dimensional Hilbert space, $\cM_{B}$ a von Neumann algebra, and $\omega_{AB}\in\cS_{\leq}(\cM_{AB})$. Then, we have that
\begin{align}
H_{\min}(A'|C)_{\omega}&=-H_{\max}(A|B)_{\omega}\ ,
\end{align}
where $\omega_{A'B'C}$ is a purification $(\pi,\cK,\ket{\xi})$ of $\omega_{AB}$ with $\cM_{A'B'}=\pi(\cM_{AB})$ the principal system, and $\cM_{C}=\pi(\cM_{A'B'})'$ the purifying system.
\end{lemma}


\section{Properties of the Conditional Von Neumann Entropy}

\begin{lemma}(Ref.~\onlinecite[Corollary 5.12 (iii)]{PetzBook})\label{lem:mono}
Let $\cM_{A},\cM_{B}$ be von Neumann algebras, $\omega_{A},\sigma_{A}\in\cP(\cM_{A})$, and let $\cE:\cM_{B}\rightarrow\cM_{A}$ be a normal, completely positive, and unital map. Then, we have that
\begin{align}
D(\omega_{A}\|\sigma_{A})\geq D(\omega_{A}\circ\cE\|\sigma_{A}\circ\cE)\ .
\end{align}
\end{lemma}

We further need a slight extension of the above lemma to sub-unital maps.

\begin{lemma}\label{lem:mono:subunital}
Let $\cM_{A},\cM_{B}$ be von Neumann algebras, $\omega_{A},\sigma_{A}\in\cP(\cM_{A})$, and let $\cE:\cM_{B}\rightarrow\cM_{A}$ be a normal, completely positive, and sub-unital map. Then, we have that
\begin{align}
D(\omega_{A}\|\sigma_{A})\geq D(\omega_{A}\circ\cE\|\sigma_{A}\circ\cE) +  D(\omega_A(\idty-\cE(\idty))\|\sigma_A(\idty-\cE(\idty)) )\ . 
\end{align}
Furthermore, if $\sigma(\idty)\leq 1$, we have that 
\begin{align}
D(\omega_{A}\|\sigma_{A})\geq D(\omega_{A}\circ\cE\|\sigma_{A}\circ\cE) + \omega_A(\idty-\cE(\idty))\log \omega_A(\idty-\cE(\idty))\ . 
\end{align}
\end{lemma}

\begin{proof}
Let $\cM_A \subset \cB(\cH_A)$. According to Stinespring's dilation theorem, Ref.~\onlinecite[Theorem 4.1]{Paulsen}, there exists a Hilbert space $\cK$ together with a representation $\pi:\cM_B \rightarrow \cB(\cK)$ and a bounded operator $V:\cH_A \rightarrow \cK$ such that $\cE(a) = V^* \pi(a) V$ for all $a\in \cM_B$. Without loss of generality, we can assume that $\cK \cong \cH_A$ by taking $\cH_A$ large enough. We define now the isometry $\tilde V:\cH_A \rightarrow \cK \otimes \mathbb C^2$ by $\tilde V = V\otimes \ket 0 + \sqrt{\idty -V^*V} \otimes \ket 1$, where $\ket 0 , \ket 1$ denotes an orthonormal basis of $\mathbb  C^2$. Let us further define the projectors $P = \idty \otimes \ketbra 00$ and $P_\bot = \idty \otimes \ketbra 11$ in $\cK\otimes \mathbb C^2$. It is then easy to see that $\cE= \text{T}_{V} \circ \text{T}_P \circ \tilde\pi_A$, where $\tilde \pi_A = \pi_A \otimes \idty$ and $\text{T}_{M}$ denotes the map $a \mapsto M^* a M$ for every $M$ with suitable range.  Applying the monotonicity of the quantum relative entropy under unital completely positive maps, Lemma~\ref{lem:mono}, we find that 
\begin{align}
D(\omega_{A}\|\sigma_{A}) &\geq D(\omega_A \circ \text{T}_V \circ (\text{T}_P + \text{T}_{P_\bot})\| \sigma_A  \circ \text{T}_V \circ (\text{T}_P + \text{T}_{P_\bot}) )
\\ & = D(\omega_A \circ \text{T}_V \circ \text{T}_P \| \sigma_A  \circ \text{T}_V \circ \text{T}_P ) + D(\omega_A \circ \text{T}_V \circ \text{T}_{P_\bot} \| \sigma_A  \circ \text{T}_V \circ \text{T}_{P_\bot} )  \, ,
\end{align}
where we used in the second step that the map $\text{T}_P + \text{T}_{P_\bot}$ divides the range into two orthogonal subspaces for which the spatial derivative decays in a direct sum with respect to these orthogonal subspaces. Applying $\pi_A$ to the first term, we obtain $D(\omega_A \circ \text{T}_V \circ \text{T}_P \| \sigma_A  \circ \text{T}_V \circ \text{T}_P ) \geq  D(\omega_A \circ \cE \| \sigma_A  \circ \cE ) $ due to $\cE= \text{T}_{V} \circ \text{T}_P \circ \tilde\pi_A$ and Lemma~\ref{lem:mono}. To the second term we apply the restriction onto the subalgebra generated by $P_\bot$, which is isomorphic to $\mathbb C$, to find $D(\omega_A \circ \text{T}_V \circ \text{T}_{P_\bot} \| \sigma_A  \circ \text{T}_V \circ \text{T}_{P_\bot} )\geq D(\omega_A(\idty-\cE(\idty))\|\sigma_A(\idty-\cE(\idty)) )$ with Lemma~\ref{lem:mono}. This proves the first assertion. The second one follows from the first by using that the term $ - \omega_A(\idty-\cE(\idty))\log \sigma_A(\idty-\cE(\idty))$ is positive whenever $\sigma(\idty)\leq 1$. 
\end{proof}

\begin{lemma} (Ref.~\onlinecite[Corollary 5.12 (ii)]{PetzBook})\label{lem:petz1}
Let $\cM$ be a von Neumann algebra, and $\omega,\sigma\in\cP(\cM)$ with $\sigma\geq\gamma$. Then, we have that
\begin{align}
D(\omega\|\gamma)\geq D(\omega\|\sigma)\ .
\end{align}
\end{lemma}

\begin{lemma}(Ref.~\onlinecite[Proposition 5.1]{PetzBook})\label{lem:petz2}
Let $\cM$ be a von Neumann algebra, $\omega,\sigma\in\cP(\cM)$, and $c>0$. Then, we have that
\begin{align}
D(\omega\|c\cdot\sigma)=D(\omega\|\sigma)+ \omega(\idty)\log\frac{1}{c}\ .
\end{align}
\end{lemma}

\begin{lemma} (Ref.~\onlinecite[Corollary 5.12 (iv)]{PetzBook})\label{lem:vNconv}
Let $\cM$ be a von Neumann algebra, and let $\{\cM_{i}\}_{i\in\mathbb{N}}$ be a sequence of von Neumann subalgebras of $\cM$ such that their union is weakly dense in $\cM$. If $\omega,\sigma\in\cP(\cM)$, then the increasing sequence $D(\omega_{\cM_{i}}\|\sigma_{\cM_{i}})$ converges to $D(\omega\|\sigma)$, where $\omega_{\cM_{i}}$ denotes the restriction of $\omega$ onto the subalgebra $\cM_{i}$.
\end{lemma}

\begin{lemma}(Ref.~\onlinecite[Corollary 5.20]{PetzBook})\label{lem:petz3}
Let $\cM_{AB}=\cM_{A}\otimes\cM_{B}$ be a tensor product of von Neumann algebras, and let $\omega_{AB}\in\cS(\cM_{AB})$, $\sigma_{A}\in\cS(\cM_{A})$, as well as $\sigma_{B}\in\cS(\cM_{B})$. Then, we have that
\begin{align}
D(\omega_{AB}\|\sigma_{A}\otimes\sigma_{B})=D(\omega_{A}\|\sigma_{A})+D(\omega_{AB}\|\omega_{A}\otimes\sigma_{B})\ .
\end{align}
\end{lemma}

\begin{lemma}\label{lem:duality_vN}
Let $\cM_{AB}=\cB(\cH_{A})\otimes\cM_{B}$ with $\cH_{A}$ a finite-dimensional Hilbert space, $\cM_{B}$ a von Neumann algebra, and $\omega_{AB}\in\cS_{\leq}(\cM_{AB})$. Then, we have that
\begin{align}
H(A'|C)_{\omega}&=-H(A|B)_{\omega}\ ,
\end{align}
where $\omega_{A'B'C}$ is a purification $(\pi,\cK,\ket{\xi})$ of $\omega_{AB}$ with $\cM_{A'B'}=\pi(\cM_{AB})$ the principal system, and $\cM_{C}=\pi(\cM_{A'B'})'$ the purifying system.
\end{lemma}

\begin{proof}
Let $V$ be the Stinespring dilation of the trace map, i.e., $\tr_A (x) = V^* x \otimes \idty \,V$ for $x \in \cM_A$. By the definition of the quantum relative entropy (Definition~\ref{def:relative}), the assertion would follow from the identity
\begin{align}\label{eq:5}
\Delta( \omega_{AB}' / V \omega_B V^* ) = \Delta( V \omega_B V^* / \omega_{AB}')^{-1} = \Delta( V \omega_C V^* / \omega_{AC}' )\ ,
\end{align}
where the first equality can be found in Ref.~\onlinecite[Chapter 4]{PetzBook}. However, the commutant of $\cM_{AC}$ is exactly $\cM_B$, hence $\omega_{AC}' = \omega_B$, and similarly $\omega_{AB}' = \omega_C$. We are left to show that for $c \in \cM_C$
\begin{align}\label{eq:6}
\omega_C(r_{V \omega_B V^*}(c V(\xi') ) r_{V \omega_B V^*}(c V(\xi') )^*) = \omega_C( \tr_A [ r_{\omega_B}(c V(\xi')) r_{\omega_B}(c V(\xi'))^*] )\ ,
\end{align}
for $\xi'$ the GNS vector associated with $\omega_B$. Indeed, the state $\tr_A \otimes \omega_B$ may assumed to be faithful for $\cM_{AB}$ (otherwise restrict everything to the support of $\omega_{AB}$) from which it follows that both $\cM_{AB}$ as well as $\cM_C = \cM_{AB}'$ are faithfully represented on the associated GNS Hilbert space $\cH$. It also follows that $\cH = HS(\cH_{A})\otimes \cH_B$, where $|A|$ denotes the dimension of $A$ and $\cH_B$ is the GNS Hilbert space associated to $\cM_B$ with respect to $\omega_B$. The discussion so far implies that the linear span of vectors of the form $c V(\xi')$ is dense in $\cH$ and $r_{V \omega_B V^*}(c V(\xi') ) = c$ for $c \in \cM_C$, and with this, the left hand side of~\eqref{eq:6} becomes
\begin{align}
\omega_C(r_{V \omega_B V^*}(c V(\xi') ) r_{V \omega_B V^*}(c V(\xi') )^*) = \omega_{C}(cc^{*})\ .
\end{align}
Moreover, the linear span of vectors $b \xi'$, $b \in \cM_B$ is dense in $\cH_B$, and we have
\begin{align}\label{eq:18}
r_{\omega_B}(c V(\xi')) ( b \xi' ) = b c V(\xi') =  c V(b \xi')\,
\end{align}
since the isometry $V$ just acts as tensoring with the identity in $\cM_A$ (in the Hilbert-Schmidt-picture), from which it follows that $r_{\omega_B}(c V(\xi'))^* = V^* c$. It may be checked that the action of $V^*$ on $HS(\cH_{A})$ is given by
\begin{align}\label{eq:8}
V^* = \sum_{i = 1}^{|A|} \pi_L(\bra{i}) \cdot \pi_R(\ket{i})\ ,
\end{align}
where $\pi_L$ respectively $\pi_R$ again denote the left respectively right action of some matrix on $HS_{|A|}$. The operator $\pi_{L}(\bra{i})$ is by definition an element of $\cM_C$, whereas $\pi_R(\ket{i})$ is an element of $\cM_A$. The partial trace taken on $\cM_A$ then reduces the operator $r_{\omega_B}(c V(\xi')) r_{\omega_B}(c V(\xi'))^*$ to $c c^*$, which proves the assertion.
\end{proof}



\end{document}